%% file: tokenforwarding.tex
\documentclass[11pt]{article}  

\usepackage{amssymb}
\usepackage{amsmath,amssymb}
\usepackage{amsthm}

\usepackage{epsfig,graphicx}

\usepackage[noend]{algorithmic}
\usepackage{algorithm}

\usepackage{fullpage}
\usepackage{subfigure}
\usepackage{times}
\usepackage{url}

\newif\iflong
\longtrue

\newtheorem{definition}{Definition}

\newtheorem{theorem}{Theorem}
\newtheorem{lemma}{Lemma}

\IfFileExists{../EmanueleViolaDir.txt}{\def\EmanueleViolaDir{1}}{
\def\EmanueleViolaDir{0}}

\ifnum\EmanueleViolaDir=1
\usepackage{../theomac}
\else
\usepackage{theomac}
\fi

\newtheoremWithMacro{theoremR}[lemma]{Theorem}
\newtheoremWithMacro{corollaryR}[lemma]{Corollary}
\newtheoremWithMacro{lemmaR}[lemma]{Lemma}
\newtheoremWithMacro{factR}[lemma]{Fact}

\begin{document}
\input{macros}
\begin{titlepage}

\title{\Large Global Information Sharing under Network Dynamics}
\author{Chinmoy Dutta\thanks{Twitter, San Francisco, USA. Email: {\tt
      chinmoy@twitter.com}.  This work was done while this author was
    at College of Computer and Information Science, Northeastern
    University, and was supported in part by NSF grant CCF-0845003 and
    a Microsoft grant to Ravi Sundaram.}  \and Gopal
  Pandurangan\thanks{Division of Mathematical Sciences, Nanyang
    Technological University, Singapore 637371 and Department of
    Computer Science, Brown University, Providence, RI 02912,
    USA. Email:~{\tt gopalpandurangan@gmail.com}.  Supported in part
    by the following research grants: Nanyang Technological University
    grant M58110000, Singapore Ministry of Education (MOE) Academic
    Research Fund (AcRF) Tier 2 grant MOE2010-T2-2-082, and a grant
    from the US-Israel Binational Science Foundation (BSF).}  \and
  Rajmohan Rajaraman\thanks{College of Computer and Information
    Science, Northeastern University, Boston, 02115, USA.  Email: {\tt
      rraj@ccs.neu.edu}. Supported in part by NSF grant CNS-0915985}
  \and Zhifeng Sun\thanks{Google, Seatle, USA. Email: {\tt
      austin@ccs.neu.edu}. This work was done while this author was at
    College of Computer and Information Science, Northeastern
    University, and was supported in part by NSF grant CNS-0915985}
  \and Emanuele Viola\thanks{College of Computer and Information
    Science, Northeastern University, Boston, 02115, USA.  Email: {\tt
      viola@ccs.neu.edu}. Supported by NSF grant CCF-0845003.}}

\date{}

\maketitle

\thispagestyle{empty}

\input{abstract}


\end{titlepage}

\input{intro}
\input{lowerbound}
\input{weakly_adaptive}

\input{offline}

\input{models}

\input{open}

\bibliographystyle{alpha}
\bibliography{tokenforwarding,sym_diff_sampling}

\end{document}

%% file: macros.tex

\newcommand{\junk}[1]{}
\newcommand{\eps}{\epsilon}

\newcommand{\sqb}[1]{\left[ #1 \right]} 
\newcommand{\cub}[1]{\left\{ #1 \right\} } 
\newcommand{\rb}[1]{\left( #1 \right)} 

\newcommand{\prob}[1]{\Pr \left[ #1 \right]}
\newcommand{\expect}[1]{\mathbb{E}\left[ #1 \right]}
\newcommand{\dprob}[3]{\Pr \left[ #1 ; #2 \mbox{;} #3 \right]}
\newcommand{\dsumprob}[2]{P(#1,#2)}

\newcommand{\symdiff}{\mbox{SYM-DIFF}}
\newcommand{\polylog}{\mbox{polylog}}
\newenvironment{LabeledProof}[1]{\noindent{\bf Proof of #1: }}{\qed}

\newcommand{\E}{\mathop{\mathbb E}}
\newcommand{\var}{\mathsf{Var}}
\newcommand{\ent}{\mathsf{H}}
\newcommand{\vol}{\mathsf{Vol}}


%% file: abstract.tex
\begin{abstract}
\begin{small}
We study how to spread $k$ tokens of information to every node on an
$n$-node dynamic network, the edges of which are changing at each
round. This basic {\em gossip problem} can be completed in $O(n +
k)$ rounds in any static network, and determining its complexity in
dynamic networks is central to understanding the algorithmic limits
and capabilities of various dynamic network models. Our focus is on
token-forwarding algorithms, which do not manipulate tokens in any way
other than storing, copying and forwarding them.

\vspace{0.1cm}

We first consider the {\em strongly adaptive} adversary model where in
each round, each node first chooses a token to broadcast to all its
neighbors (without knowing who they are), and then an adversary
chooses an arbitrary connected communication network for that round
with the knowledge of the tokens chosen by each node. We show that
$\Omega(nk/\log n + n)$ rounds are needed for any randomized
(centralized or distributed) token-forwarding algorithm to disseminate
the $k$ tokens, thus resolving an open problem raised
in~\cite{kuhn+lo:dynamic}. The bound applies to a wide class of
initial token distributions, including those in which each token is
held by exactly one node and {\em well-mixed} ones in which each node
has each token independently with a constant probability.

\vspace{0.1cm}

Our result for the strongly adaptive adversary model motivates us to
study the {\em weakly adaptive} adversary model where in each round,
the adversary is required to lay down the network first, and then each
node sends a possibly distinct token to each of its neighbors. We
propose a simple randomized distributed algorithm where in each round,
along every edge $(u,v)$, a token sampled uniformly at random from the
symmetric difference of the sets of tokens held by node $u$ and node
$v$ is exchanged. We prove that starting from any well-mixed
distribution of tokens where each node has each token independently
with a constant probability, this algorithm solves the $k$-gossip
problem in $O((n+k)\log n \log k)$ rounds with high probability over
the initial token distribution and the randomness of the protocol. We
then show how the above uniform sampling problem can be solved using
$\tilde O(\log n)$ bits of communication, making the overall algorithm
communication-efficient.

\vspace{0.1cm}

We next present a centralized algorithm that solves the gossip problem
for every initial distribution in $O((n + k)\log^2 n)$ rounds in the
offline setting where the entire sequence of communication networks is
known to the algorithm in advance. Finally, we present an $O(n
\min\{k, \sqrt{k \log n}\})$-round centralized offline algorithm in
which each node can only broadcast a single token to all of its
neighbors in each round.
\end{small}
\end{abstract}

%% file: intro.tex
\section{Introduction}
In a dynamic network, nodes (processors/end hosts) and communication
links can appear and disappear over time.  Modern networking
technologies such as ad hoc wireless, sensor, mobile, overlay, and
peer-to-peer (P2P) networks are inherently dynamic,
bandwidth-constrained, and unreliable. This necessitates the
development of a solid theoretical foundation to design efficient,
robust, and scalable distributed algorithms and understand the power
and limitations of distributed computation on such networks. Such a
foundation is critical to realize the full potential of these
large-scale dynamic networks.

In this paper, we study a fundamental problem of information
spreading, called {\em $k$-gossip}, on dynamic networks.  This problem
was analyzed for static networks by Topkis~\cite{topkis:disseminate},
and was first studied on dynamic networks by Kuhn, Lynch, and
Oshman~\cite{kuhn+lo:dynamic}.  In $k$-gossip (also referred to as
{\em $k$-token dissemination}), there are $k$ distinct pieces of
information (tokens) that are initially present in some nodes and the
problem is to disseminate all the $k$ tokens to all the $n$ nodes in
the network, under the bandwidth constraint that one token can go
through an edge per round, under a synchronous model of communication.
This problem is a fundamental primitive for distributed computing;
indeed, solving $n$-gossip, where each node starts with exactly one
token, allows any function of the initial states of the nodes to be
computed, assuming the nodes know $n$~\cite{kuhn+lo:dynamic}.

\junk{Indeed, determining the time complexity of gossip is central to
  understanding the power of distributed computation in dynamic
  networks as well to understanding the fundamental algorithmic
  limitations and capabilities of various models of dynamic networks.}

The dynamic network models that we consider in this paper allow an
adversary to choose an arbitrary set of communication links among the
nodes for each round, with the only constraint being that the
resulting communication graph is connected in each round. Our
adversarial models are either the same as or closely related to those
adopted in recent
studies~\cite{avin+kl:dynamic,kuhn+lo:dynamic,odell+w:dynamic,santoro}.

\junk{ Our dynamic models subsume several well-studied models --- in
  particular, stochastic evolutionary models (e.g., see
  \cite{Markovian} and the references therein) --- in distributed
  computing and networking. }

The focus of this paper is on the power of {\em token-forwarding}\/
algorithms, which do not manipulate tokens in any way other than
storing, copying, and forwarding them.  Token-forwarding algorithms
are simple and easy to implement, typically incur low overhead, and
have been widely studied (e.g, see~\cite{leightonbook,pelegbook}).  In
any $n$-node static network, a simple token-forwarding algorithm that
pipelines token transmissions up a rooted spanning tree, and then
broadcasts them down the tree completes $k$-gossip in $O(n + k)$
rounds~\cite{topkis:disseminate, pelegbook}, which is tight since
$\Omega(n+k)$ rounds is a straightforward lower bound due to bandwidth
constraints.  The central question motivating our study is whether a
linear or near-linear bound is achievable for $k$-gossip on dynamic
networks.

\junk{Unlike prior models on dynamic networks, this
model does not assume that the network eventually stops changing and
requires that the algorithms work correctly and terminate even in
networks that change continually over time.}

\junk{and nodes do not know their neighbors for the current round before
they broadcast their messages. (Note that in this model, only edges
change and nodes are assumed to be fixed.)
(By just gossip, we mean $n$-gossip.)  }

\subsection{Our results}
Our first result, in Section~\ref{sec:lower}, is a lower bound for
$k$-gossip under a worst-case model due to~\cite{kuhn+lo:dynamic},
which we call the {\em strongly adaptive adversary}\/ model. We now
define the model and then state the theorem.

\begin{definition}[Strongly adaptive adv.]
\label{def:strong} 
In each round of the {\em strongly adaptive adversary}\/ model, each
node first chooses a token to broadcast to all its neighbors (without
knowing who they are), and then the adversary chooses an arbitrary
connected communication network for that round with the knowledge of
the tokens chosen by each node.
\end{definition}

We note that the choice made by each node may depend arbitrarily on
the tokens held by that and other nodes.  Hence this model allows for
both distributed and centralized algorithms.

\begin{theoremR}[\tAlgLower]
\label{thm:alg+lower} (a) Any randomized token-forwarding algorithm 
(centralized or distributed) for $k$-gossip needs $\Omega(nk/\log n +
n)$ rounds in the strongly adaptive adversary model starting from any
initial token distribution in which each of $k \le n$ tokens is held
by exactly one node. (b) In addition, the same bound holds with high
probability over an initial token distribution where each of the $n$
nodes receives each of $k \le n$ tokens independently with probability
$3/4$.
\end{theoremR}

This result resolves an open problem raised in~\cite{kuhn+lo:dynamic},
improving their lower bound of $\Omega(n \log n)$ for $k = \omega(\log
n \log \log n)$, and matching their upper bound to within a
logarithmic factor.  Our lower bound also enables a better comparison
of token-forwarding with an alternative approach based on network
coding due to ~\cite{haeupler:gossip,haeupler+k:dynamic}. Assuming the
size of each message is bounded by the size of a token, network coding
completes $k$-gossip in $O(nk/\log n + n)$ rounds for $O(\log n)$-bit
tokens, and $O(n + k)$ rounds for $\Omega(n \log n)$ bit tokens.
Thus, for large token sizes, our result {\em establishes a factor
  $\Omega(\min\{n,k\}/\log n)$ gap between token-forwarding and
  network coding}, a significant new bound on the network coding
advantage for information dissemination.\footnote{The strongly
  adaptive adversary model allows each node to broadcast one token in
  each round, and thus our bounds hold regardless of the token size.}
Furthermore, for small token and message sizes (e.g., $O(\polylog(n))$
bits), we do not know of any algorithm (network coding, or otherwise)
that completes $k$-gossip against a strongly adaptive adversary in
$o(nk/\polylog(n))$ rounds.

\junk{In a key
result,~\cite{kuhn+lo:dynamic} showed that in their adversarial model,
$k$-gossip can be solved by token-forwarding in $O(nk)$ rounds, and
any deterministic online token-forwarding algorithm needs $\Omega(n
\log k)$ rounds. They also proved an $\Omega(nk)$ lower bound for a
restricted class of token-forwarding algorithms, called
knowledge-based algorithms.  Our main result is a new lower bound that
applies to {\em any}\/ token-forwarding algorithm for $k$-gossip.

We show that every token-forwarding algorithm for the $k$-gossip
problem takes $\Omega(n + nk/\log n)$ rounds against an adversary
that, at the start of each round, knows the randomness used by the
algorithm in the round.  This also implies that any deterministic
online token-forwarding algorithm takes $\Omega(n + nk/\log n)$
rounds.  Our result applies even to centralized token-forwarding
algorithms that have a global knowledge of the token distribution.}

\smallskip

Our lower bound for the strongly adaptive adversary model
motivates us to study models which restrict the power of
the adversary and/or strengthen the capabilities of the
algorithm. We would like to restrict the adversary power
as little as possible and yet design fast algorithms.

\begin{definition}[Weakly adaptive adv.]
\label{def:weak}
In each round of the {\em weakly adaptive adversary}\/ model, the
adversary is required to lay down the communication network first,
before the nodes can communicate. Hence nodes get to know their
neighbors and thus each node can send a possibly distinct token to
each of its neighbors.  Note that the adversary still has full control
of the topology in each round.
\end{definition}

We propose a simple protocol which we call the {\em symmetric difference}
(\symdiff) protocol.

\begin{definition}[\symdiff~protocol]
The protocol \symdiff~works as follows: in each round, independently
along every edge $(u,v)$, sample a token $t$ uniformly at random from
the symmetric difference (i.e., XOR) of the sets of tokens held by
node $u$ and node $v$ at the start of the round. Then the node that
holds $t$ sends it to the other node.
\end{definition}

Our second main result, in Section~\ref{sec:sym_diff_analysis}, shows
that in the weakly adaptive model, the \symdiff~protocol beats the
lower bound for mixed starting distribution of Theorem
\ref{thm:alg+lower}.

\begin{theoremR}[\trandsymdiff]
\label{thm:rand_sym_diff} Starting from any well-mixed distribution of tokens
where each of the $n$ nodes has each of the $k$ tokens independently
with a positive constant probability, the \symdiff~protocol completes
$k$-gossip in $O((n+k) \log n \log k)$ rounds with high
probability. The probability is both over the initial assignment of
tokens and the randomness of the protocol.
\end{theoremR}

\junk{We note that for this conjecture to be
true, randomization in the \symdiff~protocol is essential. Indeed, if the token
sent along an edge is selected deterministically from the symmetric difference,
as opposed to being selected at random, then we can exhibit a distribution on
$k$ tokens starting from which this protocol takes $\Omega(nk)$ rounds (see
Theorem \ref{thm:det_sym_diff}).}

A communication-efficient implementation of \symdiff~hinges on the
communication complexity of sampling a uniform element from the
symmetric difference of two sets. As another technical contribution,
we give an explicit, communication-efficient protocol for this task in
Section~\ref{sec:sym_diff_sampling}.

\begin{theoremR}[\tccSample]
\label{thm:sym_diff_sampling} Let Alice and Bob have two subsets $A \subseteq
[k]$ and $B \subseteq [k]$ respectively. There is an explicit, private-coin
protocol to sample a random element from the symmetric difference of the two
sets, $A \oplus B := (A \setminus B) \cup (B \setminus A)$, such that the
sampled distribution is statistically $\epsilon$-close to the uniform
distribution on $A \oplus B$ and the protocol uses $O(\log^{3/2} (k/\epsilon))$
bits of communication.
\end{theoremR}

A recent improvement on pseudorandom generators for combinatorial
rectangles~\cite{GMRTV12} implies an improvement in the communication in
Theorem \ref{thm:sym_diff_sampling} to $\tilde O(\lg k/\eps)$.
We also note that for \symdiff~to be
communication-efficient it is important that we work with symmetric
difference as opposed to set difference, which might have looked a
natural choice.  This is because Theorem~\ref{thm:sym_diff_sampling}
becomes false if we replace symmetric difference $A \oplus B$ with set
difference $A \setminus B$. For the latter, communication $\Omega(k)$
is required, due to the lower bounds for
disjointness~\cite{KalyanasundaramS92,Razborov92}.



\junk{ We propose \symdiff, a simple randomized distributed algorithm where in
each round, along every edge $(u,v)$, a token selected uniformly at random from
the symmetric difference of the sets of tokens held by node $u$ and node $v$ is
exchanged.
\begin{theoremR}[\tSymDiffProtocol]
\label{thm:rand_sym_diff}
In the weakly adaptive adversary model, from an initial distribution
of tokens where each node has each token with probability
$\frac{3}{4}$ independently, randomized \symdiff\ completes $k$-gossip
in $O(n \log n \log k)$ rounds with high probability.
\end{theoremR}

A critical component in \symdiff\ is the uniform selection of a token
from the symmetric difference of the sets of tokens held by two
neighboring nodes.  We show that that the uniform selection problem
can be solved with $O(\log^{1.5} n)$ bits of communication, making the
overall algorithm communication-efficient.

\begin{theoremR}[\tSymDiffSampling]
\label{thm:sym_diff_sampling}
Let Alice and Bob have two subsets $A \subseteq [n]$ and $B \subseteq
[n]$ respectively. There is an explicit protocol to sample a random
element from the symmetric difference of the two sets, $(A \setminus
B) \cup (B \setminus A)$, such that the sampled distribution is
statistically $\epsilon$-close to the uniform distribution on the
symmetric difference and the protocol uses $O(\log^{\frac{3}{2}} n +
\log^{\frac{3}{2}} (\frac{1}{\epsilon}))$ bits of communication. The
protocol outputs $\perp$ if the symmetric difference is empty.
\end{theoremR}
}

\junk{ We prove that starting from any {\em well-mixed}\/ distribution of
tokens, this algorithm solves $k$-gossip against a weakly adaptive adversary in
$O((n+k)\log n)$ rounds with high probability. We then show how the above
uniform selection problem can be solved with $O(\log^{1.5} n)$ bits of
communication, making the overall algorithm communication-efficient. }


\smallskip
Although we have only been able to establish the efficiency of the
\symdiff~protocol starting from well-mixed distributions as in Theorem
\ref{thm:rand_sym_diff}, we conjecture that in fact \symdiff~is
efficient starting from any token distribution.  A priori, however, it
is unclear if there is {\em any}\/ token-forwarding algorithm that
solves $k$-gossip in $\tilde{O}(n+k)$ rounds even in an offline
setting, in which the network can change arbitrarily each round, but
the entire evolution is known to the algorithm in advance.  Our next
result, in Section~\ref{sec:multiport}, resolves this problem.

\begin{definition}[Offline algorithm]
\label{def:offline}
An {\em offline algorithm}\/ for $k$-gossip takes as input an initial
token distribution and a sequence of $nk$ graphs $G_1$, \ldots,
$G_{nk}$, where $G_t$ represents the communication network in round
$t$.  The output of the algorithm is a schedule that specifies, for
each $t$, each edge $e$ of $G_t$, a token (if any) sent along $e$ in
each direction.  The {\em length}\/ of the schedule is the largest $t$
for which a token is sent on any edge in round $t$.
\end{definition}

\begin{theoremR}[\tOfflineMultiport]
\label{thm:offline_multiport}
There is a polynomial-time randomized offline algorithm that returns,
for every $k$-gossip instance, a schedule of length $O((n+k)\log^2 n)$
with high probability.
\end{theoremR}

\junk{
\begin{itemize}
\item
We present a polynomial-time algorithm that given any $k$-gossip
instance on an $n$-node dynamic network, computes an offline schedule
for solving the instance in $O((n+k)\log^2 n)$ rounds, with high
probability.
\end{itemize}
}

Like \symdiff, the schedule returned by the above offline algorithm
allows each node to send a possibly distinct token to each of its
neighbors in each round. However, in some applications, e.g., wireless
networks, the preferred mode of communication is broadcast.  Hence, we
also consider offline broadcast schedules where each node can only
broadcast a single token to all of its neighbors in each round and
show the following result in Section~\ref{sec:upper}.

\begin{theoremR}[\tOfflineBroadcast]
\label{thm:offline_broadcast}
There is a polynomial-time randomized offline algorithm that returns,
for every $k$-gossip instance, a broadcast schedule of length $O(n
\min\{k, \sqrt{k \log n}\})$, with high probability.
\end{theoremR}

\junk{
\item
We present a polynomial-time algorithm that given any $k$-gossip
instance on an $n$-node dynamic network, computes an offline broadcast
schedule for solving the instance in $O(n \min\{k, \sqrt{k \log n}\})$
rounds, with high probability.
\end{itemize}
}
\junk{
We present a polynomial-time centralized algorithm that solves the
$k$-gossip problem in the offline setting of an $n$-node dynamic
network in $O(\min\{nk, n \sqrt{k \log n}\})$ rounds with high
probability.  We also present a polynomial-time centralized
token-forwarding algorithm that solves the $k$-gossip problem in the
offline setting in $O(n^\eps)$ times the optimal number of rounds, for
any $\eps > 0$, assuming the algorithm is allowed to transmit $O(\log
n)$ tokens per round.

Our upper bounds show that in the offline setting, token-forwarding
algorithms can achieve a time bound that is within $O(\sqrt{k\log n})$
of the information-theoretic lower bound of $\Omega(n + k)$, and that
we can approximate the best token-forwarding algorithm to within a
$O(n^\eps)$ factor, with logarithmic extra bandwidth per edge.
}

\subsection{Related work}
Information spreading (or dissemination) in networks is a fundamental
problem in distributed computing and has a rich literature. The
problem is generally well-understood on static networks, both for
interconnection networks~\cite{leighton:book} as well as general
networks~\cite{lynch:distributed,pelegbook,attiya+w:distributed}.  In
particular, the $k$-gossip problem can be solved in $O(n + k)$ rounds
on any $n$-node static network~\cite{topkis:disseminate}.  There also
have been several papers on broadcasting, multicasting, and related
problems in static heterogeneous and wireless networks (e.g.,
see~\cite{alon+blp:radio,bar-yehuda+gi:radio,bar-noy+gns:multicast,clementi+mps:radio}).

Dynamic networks have been studied extensively over the past three
decades.  Early studies focused on dynamics that arise when edges or
nodes fail.  A number of fault models, varying according to extent and
nature (e.g., probabilistic vs.\ worst-case) of faults allowed, and
the resulting dynamic networks have been analyzed (e.g.,
see~\cite{attiya+w:distributed,lynch:distributed}).  There have been
several studies that constrain the rate at which changes occur, or
assume that the network eventually stabilizes (e.g.,
see~\cite{afek+ag:dynamic,dolev:stabilize,gafni+b:link-reversal}).

There also has been considerable work on general dynamic networks.
Early studies in this area
include~\cite{afek+gr:slide,awerbuch+pps:dynamic}, which introduce
building blocks for communication protocols on dynamic networks.
Another notable work is the local balancing approach
of~\cite{awerbuch+l:flow} for solving routing and multicommodity flow
problems on dynamic networks, which has also been applied to
multicast, anycast, and broadcast problems on mobile ad hoc
networks~\cite{awerbuch+bbs:route,awerbuch+bs:anycast,jia+rs:adhoc}.
To address highly unpredictable network dynamics, stronger adversarial
models have been studied
by~\cite{avin+kl:dynamic,odell+w:dynamic,kuhn+lo:dynamic} and others;
see the recent survey of \cite{santoro} and the references therein.
Unlike prior models on dynamic networks, these models and ours do not
assume that the network eventually stops changing; the algorithms are
required to work correctly and terminate even in networks that change
continually over time.  The recent work of \cite{clementi-podc12},
studies the flooding time of {\em Markovian} evolving dynamic graphs,
a special class of evolving graphs.
 \junk{The model of~\cite{kuhn+lo:dynamic} allows for a much stronger
   adversary than the ones considered in past
   work~\cite{awerbuch+l:flow,awerbuch+bbs:route,awerbuch+bs:anycast}.
   There also have been other prior models for dynamic networks
   similar in spirit to the model of , In addition to the $k$-gossip
   problem,~\cite{kuhn+lo:dynamic} considers the related problem of
   counting, and generalizes its results to the $T$-interval
   connectivity model, which includes the constraint that any interval
   of $T$ rounds has a stable connected spanning subgraph. } The
 survey of~\cite{kuhn-survey} summarizes recent work on dynamic
 networks.  We also note that our model and the ones we have discussed
 thus far only allow edge changes from round to round; the recent work
 of \cite{p2p-soda} studies a dynamic network model where both nodes
 and edges can change in each round.  \junk{They show that stable
   almost-everywhere agreement can be efficiently solved in such
   networks even in adversarial dynamic settings. }

\junk{ .  Local balancing algorithms, which continually balance the
  packet queues across each edge of the network and drain packets at
  their destination,

It has been shown that assuming the queues at the nodes can hold
enough packets, the local balancing approach can achieve throughput
that is arbitrarily close to the optimal achievable by any offline
algorithm.

Modeling general dynamic networks has gained renewed attention with
the recent advent of heterogeneous networks composed out of ad hoc,
and mobile devices.  

The work of~\cite{avin+kl:dynamic} studied the {\em cover time} of random walks
in a dynamic network controlled by an adversary that is oblivious to
the random choices made by the nodes (this is a weaker model than the
adaptive models considered in this paper). 
}

 \junk{As in the Kuhn
  et al. model, the algorithms in \cite{p2p-soda} will work and
  terminate correctly even when the network keeps continually
  changing.  We note that there has been considerable prior work in
  dynamic P2P networks (see \cite{p2p-soda, p2p-focs} and the
  references therein) but these don't assume that the network keeps
  continually changing over time.}

Recent work of~\cite{haeupler:gossip,haeupler+k:dynamic} presents
information spreading algorithms based on network
coding~\cite{ahlswede+cly:coding}.  As mentioned earlier, one of their
important results is that the $k$-gossip problem on the adversarial
model of~\cite{kuhn+lo:dynamic} can be solved using network coding in
$O(n+k)$ rounds assuming the token sizes are sufficiently large
($\Omega(n\log n)$ bits). For further references to using network
coding for gossip and related problems, we refer to
~\cite{haeupler:gossip,haeupler+k:dynamic,avin1,avin2,deb+mc:coding,shah}
and the references therein.

As we show in Section~\ref{sec:upper}, the problem of finding an
optimal broadcast schedule in the offline setting reduces to the
Steiner tree packing problem for directed
graphs~\cite{cheriyan+s:steiner}.  This problem is closely related to
the directed Steiner tree problem (a major open problem in
approximation
algorithms)~\cite{charikar+ccdgg:steiner,zosin+k:steiner} and the gap
between network coding and flow-based solutions for multicast in
arbitrary directed networks~\cite{agarwal+c:coding,sanders+et:flow}.

Finally, we note that a number of recent studies solve $k$-gossip and
related problems using {\em gossip-based}\/ processes, in which each
node exchanges information with a small number of randomly chosen
neighbors in each round,
e.g., see ~\cite{berenbrink+ceg:gossip,demers,kempe1,chen-spaa,karp,shah,boyd}
and the references therein.  All these studies assume a static
communication network, and do not apply directly to the models
considered in this paper.

%% file: lowerbound.tex
\section{Lower bound for the strongly adaptive adversary model}
\label{sec:lower} 

In this section, we prove Theorem \ref{thm:alg+lower}. 
We first define the adversary used in the proof of Theorem
\ref{thm:alg+lower}. 


\smallskip
\noindent {\em Adversary:} The strategy of the adversary is simple.
We use the notion of {\em free edge}\/ introduced
in~\cite{kuhn+lo:dynamic}.  In a given round $r$, we call an edge
$(u,v)$ \emph{free} \junk{
\footnote{The notion of a free edge is borrowed from
  \cite{kuhn+lo:dynamic}.  Apart from this notion, our proof technique
  is different from that of \cite{kuhn+lo:dynamic} where only a an
  $\Omega(n\log n)$ lower bound is shown.}} if at the start of the
round, $u$ has the token that $v$ broadcasts in the round and $v$ has
the token that $u$ broadcasts in the round; an edge that is not free
is called {\em non-free}.
Thus, if $(u,v)$ is a free edge in a particular round, neither $u$ nor
$v$ can gain any new token through this edge in the round. Since we
are considering a strong adversary model, at the start of each round,
the adversary knows for each node $v$, the token
that $v$ will broadcast in that round.  In round $r$, the adversary
constructs the communication graph $G_r$ as follows. First, the
adversary adds all the free edges to $G_r$. Let $C_1,C_2,\dots,C_l$
denote the connected components thus formed. The adversary then
guarantees the connectivity of the graph by selecting an arbitrary
node in each connected component and connecting them in a line. 
Figure\ref{fig:adversary} illustrates the construction.

The network $G_r$ thus constructed has exactly $l - 1$ non-free edges,
where $l$ is the number of connected components formed by the free
edges of $G_r$.  If $(u,v)$ is a non-free edge in $G_r$, then $u$, $v$
will gain at most one new token each through $(u,v)$. We refer to this
exchange on a non-free edge as a {\em useful token exchange}.

\smallskip
Our proof proceeds as follows. First, we show that with high
probability over the initial assignment of tokens, in every round
there are at most $O(\lg n)$ useful token exchanges. Then we note
that, again with high probability over the initial assignment of
tokens, overall $\Omega(nk)$ useful token exchanges must occur for the
protocol to complete.

\smallskip


\begin{definition}
We say that a sequence of nodes $v_1,v_2,\ldots,v_k$ is {\em
  half-empty}\/ in round $r$ with respect to a sequence of tokens
$t_1,t_2,\ldots,t_k$ if the following condition holds at the start of
round $r$: for all $1 \le i,j \le k$, $i \neq j$, either $v_i$ is
missing $t_j$ or $v_j$ is missing $t_i$.  We then say that $\langle
v_i \rangle$ is half-empty with respect to $\langle t_i \rangle$ and
refer to the pair $(\langle v_i \rangle, \langle t_i \rangle)$ as a
half-empty configuration of size $k$.
\end{definition}

\vspace{-5pt}
\begin{figure*}[ht]
\begin{center}
\includegraphics[width=4in]{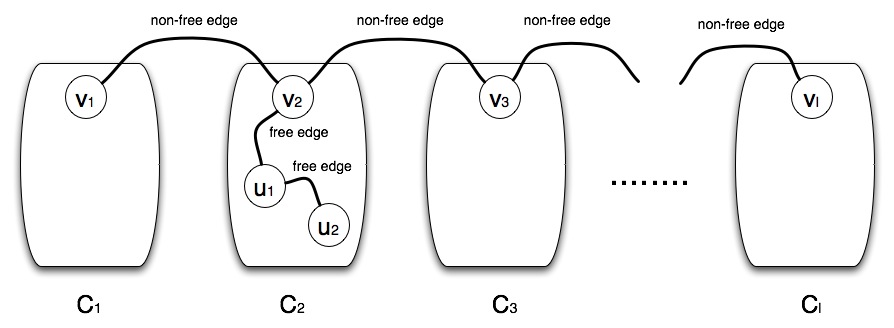}
\caption{\small{The network constructed by the adversary in a particular
  round.  Note that if node $v_i$ broadcasts token $t_i$, then the
  $\langle v_i \rangle$ forms a half-empty configuration with respect
  to $\langle t_i \rangle$ at the start of this round.}}
\label{fig:adversary}
\end{center}
\vspace{-1cm}
\end{figure*}

\junk{To prove the lower bound, we consider a special starting state, where
each node $u$ has token $t_i$ with probability $3/4$ for all $i$. Then
we look for certain structure in this starting state, and argue that
with such structure the number of useful token exchanges is $O(\log
n)$ with high probability. Furthermore, we argue that such structure
will always exist in the following rounds of communication. That
means, the number of useful token exchanges will be $O(\log n)$ in
every round. Since each node $u$ has any token with probability $3/4$
at the starting point, the expected number of missing token is $n\cdot
k/4 = \Omega(kn)$. In fact, we argue the number of missing token is
$\Omega(kn)$ with high probability using Chernoff bound. Thus, this
will imply any centralized deterministic algorithm runs in
$\Omega(kn/\log n)$ rounds.

First, we draw the connection between the number of useful token
exchanges and the existence of certain structure in Lemma
\ref{lem:free+edge}. Then in Lemma \ref{lem:alg+lower} we show the
$\Omega(n + nk/\log n)$ lower bound if the token dissemination process
starts from the state where each node $u$ has token $t_i$ with
probability $3/4$ for all $i$. Lastly, in Theorem
\ref{thm:alg+lower.single} we prove our lower bound.
\begin{lemma}
\label{lem:free+edge}
If $m$ or more useful token exchanges occur, then there exist $m/2+1$
nodes $v_1,v_2,\dots,v_{m/2+1}$ and $m/2+1$ tokens
$t_1,t_2,\dots,t_{m/2+1}$ ($v_i$ broadcasts token $t_i$) such that,
for all $i<j$, either $v_i$ is missing $t_j$ or $v_j$ is missing
$t_i$.
\end{lemma}
}

\begin{lemma}
\label{lem:free+edge}
If $m$ useful token exchanges occur in round $r$, then there exists a
half-empty configuration of size at least $m/2 + 1$ at the start of
round $r$.
\end{lemma}
\begin{proof}
Consider the network $G_r$ in round $r$.  Each non-free edge can
contribute at most 2 useful token exchanges.  Thus, there are at least
$m/2$ non-free edges.  Based on the adversary we consider, no useful
token exchange takes place within the connected components induced by
the free edges. Useful token exchanges can only happen over the
non-free edges between connected components. This implies there are at
least $m/2+1$ connected components in the subgraph of $G_r$ induced by
the free edges.  Let $v_i$ denote an arbitrary node in the $i$th
connected component in this subgraph, and let $t_i$ be the token
broadcast by $v_i$ in round $r$.  For $i \neq j$, since $v_i$ and
$v_j$ are in different connected components, $(v_i,v_j)$ is a non-free
edge in round $r$; hence, at the start of round $r$, either $v_i$ is
missing $t_j$ or $v_j$ is missing $t_i$.  Thus, the sequence $\langle
v_i \rangle$ of nodes of size at least $m/2 + 1$ is half-empty with
respect to the sequence $\langle t_i \rangle$ at the start of round
$r$.
\end{proof}

An important point to note about the definition of a half-empty
configuration is that, in a given round, it only depends on the tokens
held by the nodes; it is independent of the tokens that the nodes
broadcast.  This allows us to prove the following easy lemma that
shows a monotonicity property of half-empty configurations.
\begin{lemma}[Monotonicity Property]
\label{lem:monotone}
If a sequence $\langle v_i \rangle$ of nodes is half-empty with
respect to $\langle t_i \rangle$ at the start of round $r$, then
$\langle v_i \rangle$ is half-empty with respect to $\langle t_i
\rangle$ at the start of round $r'$ for any $r' \le r$.  Hence, the
size of the largest half-empty configuration cannot increase with the
increase in the number of rounds.
\end{lemma}
\begin{proof}
The lemma follows by noting that if a node $v_i$ is missing a token
$t_j$ at the start of round $r$, then $v_i$ is missing token $t_j$ at
the start of every round $r' < r$.
\end{proof}

Lemmas~\ref{lem:free+edge} and~\ref{lem:monotone} suggest that if we
can identify a token distribution in which all half-empty
configurations are small, we can guarantee small progress in each
round.  We now show that a well-mixed distribution satisfies the
desired property, establishing part (b) of the theorem. \junk{Note
  that the "high probability" bound in the following theorem is with
  respect to the random choices made in the initial token placement,
  and thus shows, by the probabilistic method, the lower bound holds
  for most of the initial token distributions.}

\begin{proof}[Proof of Theorem \ref{thm:alg+lower}(b)]
We first note that if the number of tokens $k$ is less
than $100 \log n$, then the $\Omega(n + nk/\log n)$ lower
bound is trivially true because even to disseminate one
token on a line it takes $\Omega(n)$ rounds\footnote{ The choice of
  the constant 100 here is arbitrary; we have not optimized the choice
  of constants in the proof.}. Thus, in the
following proof, we focus on the case where $k \ge 100
\log n$.

Let $E_l$ denote the event that there exists a half-empty
configuration of size $l$ at the start of the first round.  For $E_l$
to hold, we need $l$ nodes $v_1, v_2, \dots, v_l$ and $l$ tokens
$t_1, t_2, \dots, t_l$ such that for all $i \neq j$ either $v_i$ is
missing $t_j$ or $v_j$ is missing $t_i$. For a pair of nodes $u$ and
$v$, by union bound, the probability that $u$ is missing $t_v$ or $v$
is missing $t_u$ is at most $1/4+1/4 = 1/2$. Thus, the probability of
$E_l$ can be bounded as follows.
\[ \prob{E_l} \le {n \choose l} \cdot \frac{k!}{(k-l)!} \cdot \rb{\frac{1}{2}}^{l \choose 2} \le  n^l \cdot k^l \frac{1}{2^{l(l-1)/2}} \le \frac{2^{2l\log n}}{2^{l(l-1)/2}}.\]
In the above inequality, ${n \choose l}$ is the number of ways of
choosing the $l$ nodes that form the half-empty configuration,
$k!/(k-l)!$ is the number of ways of assigning $l$ distinct tokens,
and $(1/2)^{{l \choose 2}}$ is the upper bound on the probability for
each pair $i \neq j$ that either $v_i$ is missing $t_j$ or $v_j$ is
missing $t_i$.  For $l \geq 5 \log n$, $\prob{E_l} \le 1/n^2$.  Thus,
the largest half-empty configuration at the start of the first round,
and hence at the start of {\em any} round (by Lemma
\ref{lem:monotone}), is of size at most $5 \log n$ with probability at
least $1 - 1/n^2$.  By Lemma~\ref{lem:free+edge}, we thus obtain that
the number of useful token exchanges in each round is at most $10 \log
n$, with probability at least $1 - 1/n^2$.

\junk{Now we argue, in each following rounds, the number of useful token
exchanges is also no more than $2(l-1)$ with high probability, where $l\ge
5\log n$. If there are $2(l-1)$ or more useful token exchanges in
round $r$ where $r>1$, then by Lemma \ref{lem:free+edge} there exist
$l$ nodes $v_1,v_2,\dots,v_l$ and $l$ tokens $t_1,t_2,\dots,t_l$ such
that for all $i \neq j$ either $v_i$ is missing $t_j$ or $v_j$ is
missing $t_i$. Then at the beginning of round 1, the condition that
for all $i \neq j$ either $v_i$ is missing $t_j$ or $v_j$ is missing
$t_i$ still holds, and this can happen only with probability $1/n^2$.}

Let $M_i$ be the number of tokens missing at node $i$ in the initial
distribution. Then $M_i$ is a binomial random variable with
$\expect{M_i} = k/4$.  By a Chernoff bound, the probability that node
$i$ misses at most $k/8$ tokens is
\[ \prob{M_i \le \frac{k}{8}} = \prob{M_i \le \rb{1 - \frac{1}{2}} \cdot \expect{M_i}} \le e^{-\frac{\expect{M_i}\rb{\frac{1}{2}}^2}{2}} = e^{-\frac{k}{32}}.\]
Thus, the total number of tokens missing in the initial
distribution is at least $n \cdot k/8 = \Omega(kn)$ with
probability at least $1 - n/e^{\frac{k}{32}} \ge 1 -
1/n^2$ ($k \ge 100 \log n$).  Since the number of useful
tokens exchanged in each round is at most $10 \log n$,
the number of rounds needed to complete $k$-gossip is
$\Omega(kn /\log n)$ with high probability. 
\end{proof}

Part (b) of Theorem~\ref{thm:alg+lower} does not apply to some natural
initial distributions, such as one in which each token resides at
exactly one node.  When starting from a distribution in this class,
though there are far fewer tokens distributed initially, the argument
above does not rule out the possibility that an algorithm avoids the
problematic configurations that arise in the proof.  Part (a) of
Theorem~\ref{thm:alg+lower} extends the lower bound to this class of
distributions. The main idea of the proof is showing that a reduction
exists (via the probabilistic method) to an initial well-mixed
distribution of Theorem~\ref{thm:alg+lower}.

\junk{\begin{theorem}
\label{thm:lower.single}
Starting from any distribution in which each token starts at exactly one node,
any online token-forwarding algorithm for $k$-gossip needs $\Omega(n +
nk/\log n)$ rounds against a strong adversary.
\end{theorem}
}

\begin{lemma}
\label{lem:alg+lower.single}
From any distribution in which each token starts at exactly one node
and no node has more than one token, any online token-forwarding
algorithm for $k$-gossip needs $\Omega(kn/\log n)$ rounds against a
strong adversary.
\end{lemma}
\begin{proof}
We consider an initial distribution $C$ where each token is at exactly
one node, and no node has more than one token. Let $C^*$ be an initial
token distribution in which each node has each token independently
with probability $3/4$.  By Theorem~\ref{thm:alg+lower}, any online
algorithm starting from distribution $C^*$ needs $\Omega(kn/\log n)$
rounds with high probability.

We construct a bipartite graph on two copies of $V$, $V_1$ and
$V_2$. A node $v \in V_1$ is connected to a node $u \in V_2$ if in
$C^*$ $u$ has all the tokens that $v$ has in $C$.  We first show,
using Hall's Theorem, that this bipartite graph has a perfect matching
with very high probability.  Consider a set of $m$ nodes in $V_2$. We
want to show their neighborhood in the bipartite graph is of size at
least $m$. We show this condition holds by the following 2 cases. If
$m < 3n/5$, let $X_i$ denote the neighborhood size of node $i$. We
know $\expect{X_i} \ge 3n/4$. Then by Chernoff bound
\[ \prob{X_i < m} \le \prob{X_i < 3n/5} \le e^{-\frac{\rb{1/5}^2 \expect{X_i}}{2}} = e^{-\frac{3n}{200}}.\]
By union bound with probability at least $1-n\cdot e^{-3n/200}$ the
neighborhood size of every node is at least $m$. Therefore, the
condition holds in the first case. If $m \ge 3n/5$, we argue that the
neighborhood size of any set of $m$ nodes from $V_2$ is $V_1$ with high
probability. Consider a set of $m$ nodes, the probability that a given
token $t$ is missing in all these $m$ nodes is $(1/4)^m$. Thus the
probability that any token is missing in all these nodes is at most
$n(1/4)^m \le n(1/4)^{3n/5}$. There are at most $2^n$ such sets. By
union bound, with probability at least $1-2^n\cdot n(1/4)^{3n/5} =
1-n/2^{n/5}$, the condition holds in the second case.

By applying the union bound, we obtain that with positive probability
(in fact, high probability), $C^*$ takes $\Omega(nk/\log n)$ rounds
and there is a perfect matching $M$ in the above bipartite graph.  By
the probabilistic method, thus both $C^*$ and $M$ exist.  Given such
$C^*$ and $M$, we complete the proof as follows.  For $v\in V_2$, let
$M(v)$ denote the node in $V_1$ that got matched to $v$.  If there is
an algorithm $A$ that runs in $T$ rounds from starting state $C$, then
we can construct an algorithm $A^*$ that runs in the same number of
rounds from starting state $C^*$ as follows. First every node $v$
deletes all its tokens except for those which $M(v)$ has in $C$. Then
algorithm $A^*$ runs exactly as $A$.  Thus, the lower bound of
Theorem~\ref{thm:alg+lower}, which applies to $A^*$ and $C^*$, also
applies to $A$ and $C$. 
\end{proof}

\begin{proof}[Proof of Theorem~\ref{thm:alg+lower}(a)]
We extend our proof in Lemma~\ref{lem:alg+lower.single} to the inital
distibution $C$ where each token starts at exactly one node, but nodes
may have multiple tokens.  We consider the following two cases.

The first case is when at least $n/2$ nodes start with some
token. This implies that $k\ge n/2$.  Let us focus on the $n/2$ nodes
with tokens. Each of them has at least one unique token. By the same
argument used in Lemma~\ref{lem:alg+lower.single}, disseminating these
$n/2$ distinct tokens to $n$ nodes takes $\Omega(n^2/\log n)$
rounds. Thus, in this case the number of rounds needed is
$\Omega(kn/\log n)$.

The second case is when less than $n/2$ nodes start with some
token. In this case, the adversary can group these nodes together, and
treat them as one super node. There is only one edge connecting this
super node to the rest of the nodes. Thus, the number of useful token
exchanges provided by this super node is at most one in each round. If
there exsits an algorithm that can disseminate $k$ tokens in
$o(kn/\log n)$ rounds, then the contribution by the super node is
$o(kn/\log n)$. And by the same argument used in
Lemma~\ref{lem:alg+lower.single} we know dissemination of $k$ tokens
to $n/2$ nodes (those start with no tokens) takes $\Omega(kn/\log n)$
rounds. Thus, the theorem also holds in this case. 
\end{proof}

%% file: weakly_adaptive.tex
\section{Upper bound in the weakly adaptive adversary model}
\label{sec:weakly_adaptive}

In this section, we first analyze the \symdiff~protocol starting from
a well-mixed distribution of tokens and prove
Theorem~\ref{thm:rand_sym_diff} (presented in
Section~\ref{sec:sym_diff_analysis}), and then show how to sample an
element from the symmetric difference of two sets efficiently in the
two-player communication complexity model (presented in
Section~\ref{sec:sym_diff_sampling}). However, before doing that, we
present the following lower bound that shows randomization is crucial
for the \symdiff~protocol.

\junk{The most natural protocol in this model would, perhaps, be the {\em
  set difference} (SET-DIFF) protocol: in each round, each node sends
to each of its neighbors a token that it has but the neighbor does
not. This ensures along every edge $(u,v)$, there is a token flow from
node $u$ to node $v$ as long as the set of tokens held by node $u$ is
not a subset of that held by node $v$ and vice versa. However, it is
known that it needs $\Omega(k)$ communication bits to find a token in
the set difference and thus this protocol cannot be efficient.}

\begin{theoremR}
\label{thm:det_sym_diff}
Consider the protocol DET-SYM-DIFF for $k$-gossip in the weakly
adaptive adversary model which is identical to the \symdiff~protocol
except for, in each round, the token sent along each edge $(u,v)$ is
chosen deterministically from the symmetric difference of the set of
tokens held by node $u$ and the set of tokens held by node
$v$. Starting from an initial token distribution where one node has
all the $k$ tokens and others have none, a strongly adaptive adversary
can force $\Omega(nk)$ rounds for the DET-SYM-DIFF protocol to
disseminate the $k$ tokens to the $n$ nodes.
\end{theoremR}

\iflong
\begin{proof}
Let the node $u$ start with all the tokens and nodes $v_1, \ldots,
v_{n-1}$ start with no tokens. The adversary can connect $u, v_1,
\ldots v_{n-1}$ in a line in the first round thereby guaranteeing only
node $v_1$ gets a token, say $t_1$. In the next round, the adversary
connects $u, v_2, \ldots, v_{n-1}, v_1$ in a line. In this round, node
$v_2$ and $v_{n-1}$ will both get token $t_1$.The adversary can
continue this way for $\frac{n-2}{2} + 1$ rounds, at which point all
the nodes $v_1, v_2, \ldots, v_{n-1}$ will have token $t_1$. We can
repeat this argument for all the $k$ tokens proving the lower bound of
$\Omega(nk)$.
\end{proof}
\fi
\input{sym_diff_analysis}

\input{sym_diff_sampling}

%% file: sym_diff_analysis.tex
\subsection{Analysis of \symdiff\ starting from well-mixed distributions}
\label{sec:sym_diff_analysis}

For the proof of Theorem~\ref{thm:rand_sym_diff}, we will assume that
we start from the initial token distribution where each node has each
token independently with probability $\frac{1}{2}$.  It is easy to
extend it to any positive constant probability.  We need the following
definition. We call a maximal set of nodes that holds the same set of
tokens at the start of a round $r$ to be a {\em group} for round $r$.

\begin{lemma}
\label{lem:sym_diff_groupsize}
In a token distribution where each node has each token independently
with probability $\frac{1}{2}$, the union of the set of tokens of any
$\ell$ nodes misses at most $\frac{n+k}{\ell}$ tokens with high
probability.
\end{lemma}

\iflong
\begin{proof}
There are ${n \choose \ell}$ ways of choosing $\ell$ nodes out of $n$
nodes, and ${k \choose \frac{n+k}{\ell}}$ ways of choosing
  $\frac{n+k}{\ell}$ tokens out of $k$ tokens. Thus the probability
  that the union of the set of tokens of any $\ell$ nodes misses
 more than $\frac{n+k}{\ell}$ tokens is at most
\[ {n \choose \ell} {k \choose \frac{n+k}{\ell}} \left(\frac{1}{2}\right)^{n+k}, \]
which is inverse polynomial in both $n$ and $k$.
\end{proof}
\fi

Since in any round, no token can be exchanged along an edge between
two nodes of the same group, we will consider only the edges that
connect two nodes from different groups. We call such edges {\em
  inter-group} edges for that round. In fact, we will prove the
theorem in a stronger sense where we let the adversary orient the
inter-group edges to determine the direction of token movement
along all these edges, and the token sent along each of these edges is
chosen uniformly at random from the symmetric difference conditioned
on this orientation. (The adversary must respect the condition that
there can be no token movement from a node $u$ to a node $v$ if the
set of nodes held by node $u$ is a subset of that held by node $v$.)
We define one unit of progress in a round as a node receiving a token
in that round that it did not have at the start of the round.

\begin{lemma}
\label{lem:sym_diff_progress}
With high probability, the following holds for every node $v$ and
every round $i$: If $v$ misses $m > \log n$ tokens at the start of
round $i$ and it has $d > \log k$ incoming inter-group edges in that
round, then node $v$ makes $\Omega(\min\{m,d\})$ units of progress in
round $i$. Here, the probability is over the initial token
distribution and the randomness used in the protocol.
\end{lemma}
\iflong
\begin{proof}
First we prove the claim that that for some sufficiently small
constant $\alpha < 1$, with probability $1-o(1)$, the following holds
for every node $v$ and every round $i$: If $v$ misses $m > \log n$
tokens at the start of round $i$ and it has $d > \log k$ in-neighbors
in that round, then $\alpha d$ of these neighbors each have, at the
start of round $i$, $\alpha m$ tokens that node $v$ misses. Let us
compute the probability that the claim is not true for some node $v$
in some round $i$. The $d$ inter-group in-neighbors can be chosen in
at most ${n \choose d}$ different ways and the $m$ missing tokens can
be chosen in at most ${k \choose m}$ different ways. There are at most
${d \choose \alpha d}$ ways of choosing the in-neighbors that do not
have the claimed number of missing tokens, and for each of them there
are at most ${m \choose \alpha m}$ ways of choosing which of these
tokens they miss. Thus the probability of failure is at most
\[ {n \choose d} {k \choose m} {d \choose \alpha d} {m \choose \alpha m}^{(1-\alpha)d} \left(\frac{1}{2}\right)^{(1 - \alpha)^2 md}, \]
which is $o(\frac{1}{(nk)^2})$ since $m > \log n$ and $d > \log k$ and
$\alpha$ is chosen sufficiently small. Noting that there are at most
$n$ choices for $d$ and at most $k$ choices for $m$, the claim
follows. From the above claim, the lemma follows by standard
calculations.
\end{proof}
\fi

\begin{proof}[Proof of Theorem~\ref{thm:rand_sym_diff}]
We color each of the rounds {\em red}, {\em blue}, {\em green} or {\em
  black}. If in a round, there is a node $v$ that misses less than
$\log n$ tokens and makes at least one unit of progress in that round,
we color the round red. If a round is not colored red, and there is a
node that gets a constant fraction of its missing tokens in that round
(the same fraction as in Lemma~\ref{lem:sym_diff_progress}), we color
it green. If a round is neither colored red nor colored green, we
color the round blue. \junk{It is clear from
  Lemma~\ref{lem:sym_diff_groupsize} and
  Lemma~\ref{lem:sym_diff_progress}, that with probability $1 - o(1)$,
  we will be able to color all the rounds red, green or blue till the
  protocol completes $k$-gossip. Thus we just need to bound the number
  of red, green and blue rounds.}

It is immediate that there can be at most $n \log n$ red rounds since
each of the $n$ nodes can be responsible for coloring at most $\log n$
rounds red. Similarly, there can be at most $O(n \log k)$ green rounds
since each node can be responsible for coloring at most $O(\log k)$
rounds green. Now let us turn to the blue rounds. Fix a blue round and
let there be $r$ groups in that round. Using
Lemma~\ref{lem:sym_diff_groupsize}, we infer that there are at most
$(n+k)r$ tokens missing in total at the start of this round. We also
note that there must be at least $r-1$ inter-group edges in this round
and combining this with Lemma~\ref{lem:sym_diff_progress} and the fact
that this round was not colored red or green, we infer that we make
$\Omega(\frac{r}{\log k})$ units of progress in this round.

We can label each blue round by the smallest number of groups in a
blue round seen so far. The sequence of labels is non-increasing and
let us say it starts from $s \leq n$. We divide the blue rounds in
partitions where the $i$'th partition contain those with labels in
$[s/2^{i-1},s/2^i)$. There are at most $\log n$ partitions. From the
  above argument, we see that there can be at most $O((n + k)\log k)$
  blue rounds in each partition, which implies a bound of $O((n+k)
  \log n \log k)$ for the total number of blue rounds. This completes
  the proof of the theorem.
\end{proof}

%% file: sym_diff_sampling.tex
\subsection{Uniform sampling from symmetric difference}
\label{sec:sym_diff_sampling}

We now restate and prove our result on a
communication-efficient protocol to sample from the
symmetric difference of two sets.

\tccSample

We now explain how we obtain a communication-efficient protocol to
sample from the symmetric difference $A \oplus B$ of two sets $A, B
\subseteq [k]$, proving Theorem \ref{thm:sym_diff_sampling}.

Out starting point is Nisan and Safra's protocol \cite{Nisan93} to
determine the least $i$ such that $i \in A \oplus B$. (In
\cite{Nisan93} the protocol is phrased as deciding if $A > B$, when
$A$ and $B$ are viewed as $k$-bit integers. It is easy to switch
between the two.) For uniform sampling from $A \oplus B$, our idea is
to first let the parties permute their sets according to a random
permutation $\sigma$, then run Nisan and Safra's protocol. This
results in an explicit protocol for uniform generation from $A \oplus
B$ with communication $O(\log k/\eps)$ that uses \emph{public
  coins}. A standard transformation to private coins via
\cite{Newman91} results in a protocol that is not explicit.

To obtain an explicit, private-coin protocol we derandomize the space
of random permutations $\sigma$.  The key idea is that it is
sufficient to have a distribution on permutations $\sigma$ such that,
for any set $D = A \oplus B$, any element in $D$ has roughly the same
probability of \emph{being the first element in $D$ to appear in the
  sequence} $\sigma(1), \sigma(2), \sigma(3), \ldots$. We then
construct such a space of permutations with seed length $O(\lg^{3/2}
(k/\eps))$ using Lu's pseudorandom generator for combinatorial
rectangles \cite{Lu02}
(cf.~\cite{Nis92,NiZ96,INW94,EvenGLNV98,ArmoniSWZ96,Lu02,Viola-rbd}).
Plugging a better pseudorandom generator for combinatorial rectangles
in our argument would result in a protocol for uniform sampling from
$A \oplus B$ with communication $\tilde O(\log k/\epsilon)$ and error
$\epsilon$.

As a first step, we have the following simple
derandomization of Nisan and Safra's protocol
\cite{Nisan93}, essentially from \cite{Viola-ccsum}.

\begin{lemma} \label{lemma:NisanSafraViola}
There is an explicit, private-coin protocol to determine
the least $i \in A \oplus B$, where $A, B \subseteq [k]$,
with error $\alpha$ and communication $O(\lg(k/\alpha)
\lg \lg k) = \tilde O(\lg k/\alpha)$.
\end{lemma}
\iflong
\begin{proof}[Proof sketch]
Nisan and Safra's protocol amounts to walking for $O(\lg
k/\alpha)$ on a certain binary tree. At every node, the
two parties just need to determine with error
probability, say, $1/100$ if a portion of their inputs
are different. This latter task can be achieved using
small-bias generators with public randomness $O(\lg k)$
and communication $O(1)$.\cite{NaN93,AGHP92}

The resulting protocol can be seen as a randomized
algorithm needing a one-way stream of $R := O(\lg
k/\alpha) \lg k$ random bits and using space $S := O(\lg
k/\alpha)$ to store the current node.

Nisan's space-bounded generator \cite{Nis92} can reduce
the randomness to $S \lg(R/S) = \lg(k/\alpha) \lg \lg k$
with error loss $2^{-S} = \alpha/k$.

The parties start by exchanging a seed for Nisan's
generator, and then proceed with the previous protocol.
\end{proof}
\fi

Specifically, for given $k$ and $\epsilon$ as in Theorem
\ref{thm:sym_diff_sampling} we set $d = k \log
\left(\frac{3k}{\epsilon}\right)$ and $\alpha :=
\epsilon/3kd$. Alice then picks a random seed of length
$s(k,d,\alpha)$ for a generator that fools every
combinatorial rectangle with universe size $k$ and $d$
dimensions with error $\alpha$. That is, if $X$ is the
output of the generator on a random seed, we have, for
every set $R := R_1 \times R_2 \times \cdots R_d
\subseteq [k]^d$,
$$| \Pr[X \in R] - |R|/k^d| \le \alpha.$$

Alice sends the seed to Bob.

Both Alice and Bob expand the seed into a sample $X$ of
the generator, and use $X$ to generate a permutation
$\sigma$ as follows. Let the number of distinct elements
of $[k]$ that appear in $X$ be $t$. The permutation
$\sigma$ is constructed by defining $\sigma(i)$ to be the
$i$'th distinct element of $[k]$ that appears in $X$ as
we scan it from the beginning, for $i \leq t$. For every
$i
> t$, $\sigma(i)$ is defined to be a distinct element not
appearing in $X$ in an arbitrary but deterministic way
that is fixed before the start of the protocol and both
Alice and Bob are aware of it. (For concreteness, it can
simply be to assign the elements not appearing in $X$ by
order).

To show the correctness of our protocol we need the
following lemma.

\begin{lemma}
\label{lem:permutation} Let $X \in [k]^d$ be the output
of a combinatorial rectangle generator with error $\alpha
= \epsilon/3kd$, over a uniform seed. Let $D$ be any set,
and let $j$ be any element in $D$. The probability that
$j$ appears in a coordinate of $X$ before any other
element of $D$ is $\geq \frac{1}{|D|} -
\frac{2\epsilon}{3k}$.
\end{lemma}
\iflong
\begin{proof}

We note that the desired probability is the union of
disjoint rectangles, and then apply the property of the
generator:
\begin{align*}
\lefteqn{\Pr\left[X \in \bigcup_{0 \leq t < d} ([k] \setminus D)^t \times \{j\} \times [k]^{d-t-1} \right]} && \\
& = \sum_{0 \leq t < d} \Pr\left[X \in ([k] \setminus D)^t \times \{j\} \times [k]^{d-t-1}\right]\\
& \ge  \sum_{0 \leq t < d} |([k] \setminus D)^t \times \{j\} \times [k]^{d-t-1}|/k^d - \frac{\epsilon}{3k}\\
& =  \frac{1}{k} + \left(\frac{k-|D|}{k}\right) \frac{1}{k} + \ldots + \left(\frac{k-|D|}{k}\right)^{d-1} \frac{1}{k} - \frac{\epsilon}{3k} \\
& =  \frac{1}{k} \left(1 + \left(1-\frac{|D|}{k}\right) + \ldots + \left(1-\frac{|D|}{k}\right)^{d-1} \right) - \frac{\epsilon}{3k} \\
& =  \frac{1}{|D|} \left(1 - \left(1-\frac{|D|}{k}\right)^d \right) - \frac{\epsilon}{3k} \\
& \geq  \frac{1}{|D|} \left(1 - e^{-\frac{|D|}{k} k \log(\frac{3k}{\epsilon})}\right) - \frac{\epsilon}{3k} \\
& =  \frac{1}{|D|} - \frac{1}{|D|} \left(\frac{\epsilon}{3k}\right)^{|D|} - \frac{\epsilon}{3k} \\
& \geq  \frac{1}{|D|} - \frac{2\epsilon}{3k},
\end{align*}
since $|D| \geq 1$.
\end{proof}
\fi
Now we can complete the proof of
Theorem~\ref{thm:sym_diff_sampling}.

\begin{proof}[Proof of
Theorem~\ref{thm:sym_diff_sampling}] For given
$k,\epsilon$, we set $d = k \log
\left(\frac{3k}{\epsilon}\right)$ and $\alpha :=
\epsilon/3kd$. Alice then picks a random seed of length
$s(k,d,\alpha)$.

If $\sigma$ is chosen such that every element $j \in D$
has probability $\frac{1}{|D|}$ of preceding all other
elements of $D$, then $\sigma(i^*)$ is a uniform random
element of $D$, where $i^*$ is the first position where
the permuted $A$ and $B$ differ. Using
Lemma~\ref{lem:permutation}, we immediately see that if
$\sigma$ is chosen as in the first step of the protocol,
then the distribution of $\sigma(i^*)$ is at most
$\left(\frac{2\epsilon}{3k}\right) |D| \leq
\frac{2\epsilon}{3}$-far from the uniform distribution on
$D$.

For the second part of the protocol we use Lemma
\ref{lemma:NisanSafraViola} with $\alpha := \epsilon/3$.

Overall, the sampled distribution has distance $\le 2\epsilon/3
+ \epsilon/3 = \epsilon$ from the uniform distribution on $D$.

Using the generator in \cite{Lu02} we have $s(k,d,\alpha)
= O(\lg n + \lg d + \lg^{32} 1/\alpha) = O (\lg^{3/2}
n/\epsilon)$. So overall the communication is
$O(\lg^{3/2} n/\epsilon)$.
\end{proof}

%% file: offline.tex
\section{Offline token-forwarding algorithms}
\label{sec:centralized}
We present two offline algorithms for $k$-gossip.  The first computes
an $O((n+k)\log^2 n)$-round schedule assuming that each node can send
at most one token to each neighbor in each round
(Section~\ref{sec:multiport}); the second computes an
$O(\min\{n\sqrt{k\log n}, nk\})$-round broadcast schedule assuming
that each node can broadcast at most one token to its neighbors in
each round (Section \ref{sec:upper}).  \junk{ and a bicriteria
  {$\rb{O(n^\epsilon), O(\log n)}$-approximation} algorithm in the
  broadcast model (Section \ref{sec:approx}).  } \junk{Both
  algorithms use a leveled graph constructed from the sequence of
  dynamic graphs which we call the {\em evolution graph}.
  Appendix~\ref{app:centralized} describes this construction and
  contains all omitted proofs.}

\junk{
}

\input{multiport_flow_based}
\input{flow_based}

%% file: multiport_flow_based.tex
\newcommand{\Srcs}{{\cal S}}
\newcommand{\Snks}{{\cal T}}
\newcommand{\expandedG}[1]{\widehat{G}[#1]}
\newcommand{\Sources}[1]{\overline{S}_{#1}}
\subsection{$O((n+k)\log^2 n)$-round offline schedule}
\label{sec:multiport}
In this section, we present an algorithm for computing an $O((n+k)
\log^2 n)$ round offline schedule.  Our bound is tight to within an
$O(\log^2 n)$ factor since the dissemination of any $k$ tokens to even
a single node of the network requires $\Omega(n + k)$ rounds in the
worst case.  We begin by defining the notion of an evolution graph
that facilitates the design of the offline algorithms.

\smallskip
\noindent
{\em Evolution graph}: Let $V$ be the set of nodes. Consider a dynamic
network of $l$ rounds numbered $1$ through $l$ and let $G_i$ be the
communication graph for round $i$. The evolution graph $\expandedG{l}$
for this network is a directed capacitated graph $G$ with $l+1$ levels
constructed as follows. We create $l+1$ copies of $V$ and call them
$V_0, V_1, V_2, \dots, V_{l}$. $V_i$ is the set of nodes at level $i$
and for each node $v$ in $V$, we call its copy in $V_i$ as $v_i$. For
$i = 1, \ldots, l$, level $i-1$ corresponds to the beginning of round
$i$ and level $i$ corresponds to the end of round $i$. Level $0$
corresponds to the network at the start.  There are two kinds of edges
in the graph.  First, for every node $v$ in $V$ and every round $i$,
we place an edge with infinite capacity from $v_{i-1}$ to $v_{i}$. We
call these edges {\em buffer edges} as they ensure tokens can be
stored at a node from the end of one round to the end of the next.
Second, for every round $i$ and every edge $(u,v) \in G_i$, we place
two directed edges with unit capacity each, one from $u_{i-1}$ to
$v_{i}$ and another from $v_{i-1}$ to $u_{i}$.  We call these edges
as {\em transmit edges} as they correspond to every node transmitting
a message to a neighbor in round $i$; the unit capacity ensures that
in a given round a node can transmit at most one token to each
neighbor.  Figure~\ref{fig:evolution} illustrates our construction.
\iflong
Lemma~\ref{lem:maxflow} explains the usefulness of this construction.
\fi

\begin{figure*}[ht]
\begin{center}
\includegraphics[width=5in]{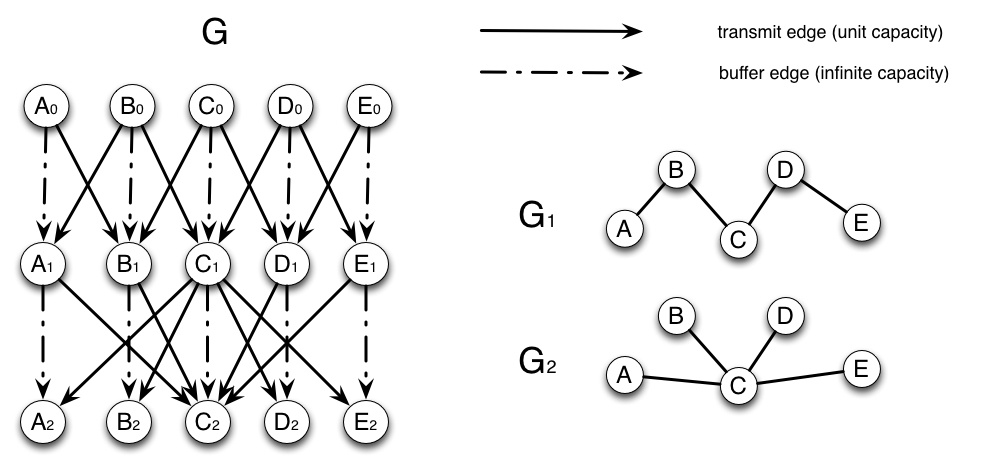}
\caption{An example of how to construct the evolution graph from a
  sequence of communication graphs.}
\label{fig:evolution}
\end{center}
\end{figure*}

\iflong
\begin{lemma}
\label{lem:maxflow}
Let $S$ be a set of source nodes, each with a subset of the $k$ tokens
and let $T$ be a subset of sink nodes.  Let $\expandedG{\ell}$ be an
evolution graph over $\ell$ rounds.  Let $P$ denote a set of
edge-disjoint paths starting from $S$ and ending at $T$.  If $P$
contains for each sink $v$ and each token $i$, a distinct path from a
source containing $i$ to $v$, then $P$ yields an $\ell$-round schedule
for disseminating the $k$ tokens to each node in $T$.
\end{lemma}
\begin{proof}
For each sink $v$, let $p_v^i$ denote the path in $P$ starting at a
source containing token $i$ and ending at $v$.  We construct a
schedule in the following natural way: for each token $i$ and sink
$v$, $p_v^i$ is the schedule by which $i$ is sent from a source to
$v$.  In particular, if $(u_t, v_{t+1})$ is in $p_v^i$, then the node
$u$ sends token $i$ to $v$ in round $t$.

We need to show that this is a feasible schedule. First we observe
that two different paths in $P$ cannot use the same transmit edge
since each such edge has unit capacity.  Next we claim by induction
that if node $v_{j}$ is in $p_v^i$, then node $v$ has token $i$ by the
end of round $j$. For $j = 0$, it is trivial since path $p_v^i$ starts
from a source that has token $i$.  For $j > 0$, if $v_{j}$ is in
$p_v^i$, then the preceding edge is either a buffer edge
$(v_{j-1},v_{j})$ or a transmit edge $(u_{j-1},v_{j})$. In the former
case, by induction node $v$ has token $i$ after round $j-1$ itself. In
the latter case, node $u$ which had token $i$ after round $j-1$ by
induction was the neighbor of node $v$ in $G_j$ and $u$ sent token $i$
in round $j$ according to $p_v^i$, thus implying node $v$ has token
$i$ after round $j$. From the above claim, we conclude that whenever a
node is asked to transmit a token in round $j$, it has the token by
the end of round $j-1$. Thus the schedule we constructed is feasible.
Since $k$ paths terminate at each of the sinks, we conclude all the
tokens reach all of the sinks after round $\ell$.
\end{proof}


Lemma~\ref{lem:maxflow} provides the foundation for the following
randomized algorithm that first gathers all tokens at a random source
node and then, in $O(\log n)$ phases, disseminates these tokens to
geometrically increasing sets of nodes, until all of the nodes have
all tokens.
\fi
\begin{algorithm}[ht!]
\caption{Computing an $O((n+k)\log^2n)$-round schedule for $k$-gossip}
\label{alg:flow_based}
\begin{algorithmic}[1]
  \REQUIRE A sequence of communication graphs $G_1, G_2, \ldots$
  \ENSURE Schedule to disseminate $k$ tokens to all nodes

  \medskip

  \STATE {\bf Gather:} Send the $k$ tokens to a node $v_0$, chosen
  uniformly at random, in $n+k$ rounds. \label{alg:gather}

  \FOR{$i$ from $0$ to $\lg n$ (Phase~$i$)} 

  \STATE Choose a set $S_i$ of $2^i$ nodes uniformly at random from
  the collection of all $2^i$-size node sets. 
  
  \STATE {\bf Flow:} Send the $k$ tokens to every node in $S_i$ using
  a maximum flow in an $O((n+k)\log n)$-round evolution graph from the
  set $\{v_0\} \cup \bigcup_{j < i}S_i$ of sources to the set $S_i$ of
  sinks.  \ENDFOR
\end{algorithmic}
\end{algorithm}

\junk{
\begin{enumerate}
\item
{\bf Gather} all the $k$ tokens at a node $v_0$, chosen uniformly at
random from $V$.

\item
For $i = 0$ to $\lg n$: (Phase~$i$)
\begin{enumerate}
\item
{\bf Disseminate} all the $k$ tokens from nodes in $\{v_0\} \cup
\bigcup_{j < i}S_i$ to nodes in $S_i$.
\end{enumerate}
\end{enumerate}
}
\iflong
We first show that the gather step can be completed in $O(n + k)$ rounds.

\begin{lemma}
\label{lem:level.flow}
Let $k$ tokens be at given source nodes and $v$ be an arbitrary
node. Then, all the tokens can be gathered at $v$ in at most $n+k$
rounds.
\end{lemma}
\begin{proof}
Following Lemma~\ref{lem:maxflow}, it suffices to show that any
evolution graph $\expandedG{n+k}$ contains $k$ edge-disjoint paths,
each starting from a source node and ending at $v$.  To prove this, we
add to $\expandedG{n+k}$ a special vertex $v_{-1}$ at level $-1$ and
connect it to every source at level $0$ by an edge of capacity
1. (Multiple edges get fused with corresponding increase in capacity
if multiple tokens have the same source.) We claim that the value of
the min-cut between $v_{-1}$ and $v_{n+k}$ is at least $k$. Before
proving this, we complete the proof of the claim assuming this.  By
the max flow min cut theorem, the max flow between $v_{-1}$ and
$v_{n+k}$ is at least $k$. Since we connected $v_{-1}$ with each of
the $k$ token sources at level $0$ by a unit capacity edge, it follows
that unit flow can be routed from each of these sources at level $0$
to $v_{n+k}$ respecting the edge capacities, establishing the desired
claim.

To prove our claimed bound on the min cut, consider any cut of the
evolution graph separating $v_{-1}$ from $v_{n+k}$ and let $S$ be the
set of the cut containing $v_{-1}$. If $S$ includes no vertex from
level $0$, we are immediately done. Otherwise, observe that if $v_{j}
\in S$ for some $0 \leq j < (n+k)$ and $v_{j+1} \notin S$, then the
value of the cut is infinite as it cuts the buffer edge of infinite
capacity out of $v_{j}$. Thus we may assume that if $v_{j} \in S$,
then $v_{j+1} \in S$. Also observe that since each of the
communication graphs $G_1, \ldots, G_{n+k}$ are connected, if the
number of vertices in $S$ from level $j+1$ is no more than the number
of vertices from level $j$ and not all vertices from level $j+1$ are
in $S$, we get at least a contribution of 1 in the value of the cut
owing to a transmit edge. But since the total number of nodes is $n$
and $v_{n+k} \notin S$, there must be at least $k$ such levels, which
proves the claim.
\end{proof}

The remainder of the proof concerns the $\lg n$ phases.  We first
establish an elementary tree decomposition lemma that is critical in
showing that there is enough capacity in any $O((n + k)\log n)$-level
evolution graph to complete each phase.

\begin{lemma}
\label{lem:tree}
For any $n$-node tree $T$ and any integer $1 \le s \le n$, there
exists an edge-disjoint partition of $T$ into subtrees $T_1$, $T_2$,
\ldots such that each $T_i$ has $\Theta(s)$ nodes, every node of $T$
is in some $T_i$, and for each $i$, at most one node in $T_i$ is in
$\bigcup_{j \neq i} T_j$.
\end{lemma}
\begin{proof}
The proof is by induction on the size of $T$.  The base case $n = 1$
is trivial.  We now consider the induction step.  Arbitrarily root the
tree $T$ at a node $r$.  For any node $v$, let $T_v$ denote the
subtree rooted at node $v$; let $n_v = |T_v$.  Thus, $n_r = n$.  Let
$v$ denote an arbitrary node such that $n_v \ge s$ and for every child
$w$ of $v$, $n_w < s$.  We first consider the case $n_v \le 2s$.  By
the induction hypothesis, there exist edge-disjoint subtrees of $T -
T_v$ such that each subtree has $\Theta(s)$ edges, every node of $T -
T_v$ is in some subtree, and any two subtrees share at most one node.
Adding $T_v$ to this collection of subtrees yields the desired claim
for $T$.

We now consider the case where $n_v > 2s$.  Here we consider two
subcases.  The first subcase is where either $v$ is the root or $|T -
T_v| \ge s$.  We partition the children of $v$ into a set $X$ of
groups such that for each group $g \in X$, $s \le 1 + \sum_{w \in g}
n_w \le 2s$.  Let $T(g)$ denote the tree $\{v\} \cup \bigcup_{w \in g}
T_w$.  All of these subtrees are edge-disjoint and any pair of
subtrees share at most one node ($v$).  If $v$ is the root, then we
have established the desired property for $T$.  Otherwise, since $|T -
T_v| \ge s$, by the induction hypothesis, there exist edge-disjoint
subtrees of $T - T_v$ such that each subtree has $\Theta(s)$ edges,
every node of $T - T_v$ is in some subtree, and for any subtree, at
most one node in the subtree is in any of the other subtrees.  Adding
the trees $T(g)$ to this collection of subtrees yields the desired
claim for $T$.  

The second subcase is where $0 < |T - T_v| < s$.  In this subcase, we
make the parent of $v$ as the child of $v$ and proceed to the first
subcase, thus establishing the desired claim and completing the
induction step.
\end{proof}

The set of sources at the start of phase $i$ is $\Sources{i} = \{v_0\}
\cup \bigcup_{j < i} S_j$.  We next place a lower bound on the size of
$\Sources{i}$.

\begin{lemma}
\label{lem:sources}
For each $i$, $0 \le i \le \lg n$, $|\Sources{i}|$ is at least
$\min\{1, 2^{i-2}\}$ with probability at least $1 - 1/n^3$;
furthermore, $\Sources{i}$ is drawn uniformly at random from the
collection of all $|\Sources{i}|$-node sets.
\end{lemma}
\begin{proof}
For $i \le \lg\lg n$, we calculate the
probability, for each $v$, that there exist more than four values of
$j$ for which $S_j$ contains $v$ as at most
\[
\binom{\lg n}{5} n \frac{1}{n^5} \le \frac{1}{n^3}.
\]
Thus, the size of the given set is at least $2^i/4 = 2^{i-2}$ with
probability at least $1 - 1/n^3$.  We now consider the case $i >
\lg\lg n$.  Let $X_v$ denote the indicator variable for node $v$ to
be in the set.  Then,
\[ E[X_v] = 1 - (1 - 1/n)\prod_{0 \le j < i}(1 - 2^j/n) \ge 1 - e^{-1/n - \sum_{j < i} 2^j/n} = 1 - e^{-2^{i}/n} \ge 4\cdot 2^{i}/(7n).\]
Thus, the expected size of the set is at least $2^{i-1}$.  Now, using
a Chernoff-type argument (e.g., by using the method of bounded
differences and invoking Azuma's inequality), we obtain the size of
the set is at least $2^{i-2}$ whp.
\end{proof}

\begin{lemma}
\label{lem:double}
Let $r \le n$ be an arbitrary integer.  Let $\Srcs$ denote a set of at
least $r/4$ sources and $\Snks$ a set of $r$ sinks, each set drawn
independently and uniformly at random from $V$.  Then, with high
probability, the evolution graph $\expandedG{\ell}$ with $\ell =
\Theta((n+k)\log n)$ contains $rk$ edge-disjoint paths, each path
starting from a source and ending at a sink, and each sink having
exactly $k$ paths ending at it.
\end{lemma}
\begin{proof}
We add a super-source having edges of capacity $rk$ to each source and
a super-sink with edges of capacity $k$ from each sink.  It thus
suffices to prove that the maximum flow from the super-source to the
super-sink is at least $rk$.  For $r \le \lg n$, we invoke
Lemma~\ref{lem:level.flow} to obtain that the maximum flow is at least
$rk$.  In the remainder of this proof, we assume $r \ge \lg n$.  We
show that with high probability, the capacity of every cut is at least
$rk$.  Note that since there are an exponentially large number of cuts
to consider, it may not be sufficient to establish a high probability
bound for each cut separately.  We address this challenge by
identifying an important property that holds for $\expandedG{\ell}$
that enables the capacity bound to hold for all cuts simultaneously.

Consider graph $G_i$ with the source and sink sets $\Srcs$ and
$\Snks$.  Recall that $\Srcs$ and $\Snks$ are drawn uniformly at
random from the collection of all $|\Srcs|$-node and $|\Snks|$-node
sets, respectively, and $\Snks'$ is an arbitrary subset of $\Snks'$ of
size $r'$.  By Lemma~\ref{lem:tree} applied to a spanning tree of
$G_i$ with parameter $s = (n\log n)/r$, there exist edge-disjoint
subtrees $T_i^1$, $T_i^2$, \ldots, each having $\Theta(s)$ edges from
the spanning tree, and together containing all of the nodes in $V$.
Furthermore, for each $T_i^j$, at most one of its nodes is present in
the other subtrees.  Since $\Srcs$ and $\Snks$ are drawn at random and
have are of size at least $r/4$ and equal to $r$, respectively, it
follows from a standard Chernoff bound that each of these subtrees has
$\Omega(\log n)$ (resp., $\Theta(\log n)$), nodes from $\Srcs$ (resp.,
$\Snks$) whp.  In the remainder of the proof, we thus assume that the
preceding property holds for each of the graphs in the
$\Theta((n+k)\log n)$ levels of $\expandedG{\ell}$.

We now argue that every cut $C = (\Srcs, \Snks)$ of $\expandedG{\ell}$ has
capacity at least $rk$.  If any of the sources in $\Srcs$ is separated
from the super-source, then the capacity of the cut is at least $rk$
since the capacity of the edge connecting the super-source to any
source is $rk$.  So in the remainder, we assume that all nodes in $S$
are on the same side of the cut as the super-source.  Let $\Snks'$
denote the set of sinks that are separated from the super-source in
$C$; let $r' = |\Snks'|$.  All of the edges from $\Snks - \Snks'$ to
the super-sink cross $C$ and have a total capacity of $(r - r')k$.  It
thus remains to show that the total capacity of the edges crossing the
cut in the intermediate levels $1$ through $t$ is at least $r'k$.

Let $V_i$ denote the set of nodes in level $i$ that are in $\Srcs$.
Since every parallel edge has infinite capacity, we have $V_{i+1}
\supseteq V_i$.  Since each $V_i$ is of size at most $n$, there are at
least $t - n$ levels such that $V_{i+1} = V_i$.  For any such level
$i$, $C$ includes all edges that separate $\Srcs$ from $\Snks'$ in the
graph $G_i$.  By the property established above, there exist
edge-disjoint partition of a spanning tree of $G_i$ that such that
each tree in the partition contains $\Theta(\log n)$ nodes from both
$\Srcs$ and $\Snks$.  Therefore, for any arbitrary subset $\Snks'$ of
size $r'$, we can find $\Omega(r'/\log n)$ edges that separate
$\Snks'$ from $\Srcs$.  For the number of levels exceeding
$\Omega(k\log n)$, it then follows that the total capacity of the
edges crossing the cut in the intermediate levels is at least $r'k$.
This establishes the desired lower bound on the capacity of the cut,
completing the proof of the lemma.
\end{proof}

\tOfflineMultiport
\begin{proof}
By Lemma~\ref{lem:level.flow}, the gather step completes in $O(n + k)$
rounds.  We now argue that each phase completes in $O((n + k)\log n)$
rounds whp.  By Lemma~\ref{lem:sources}, the number of sources at the
start of phase $i$ is at least $2^{i-2}$ whp.  By
Lemmas~\ref{lem:maxflow} and~\ref{lem:double}, the number of rounds
needed for phase $i$ is $O((n + k)\log n)$ whp.  Since the number of
phases is $\lg n$, the statement of the theorem follows.
\end{proof}
\fi

%% file: flow_based.tex
\newcommand{\bexpandedG}[1]{\widetilde{G}[#1]}

\subsection{An $O(\min\{n\sqrt{k\log n}, nk\})$-round broadcast schedule}
\label{sec:upper}
\iflong We extend the notion of the evolution graph to the broadcast
model.  The primary difference is the addition of a new level of nodes
and edges for every round that enforces the broadcast constraint.

\smallskip
\noindent
{\em Evolution graph}: Let $V$ be the set of nodes. Consider a dynamic
network of $l$ rounds numbered $1$ through $l$ and let $G_i$ be the
communication graph for round $i$. The evolution graph for this
network is a directed capacitated graph $\bexpandedG{2l+1}$ with
$2l+1$ levels constructed as follows. We create $2l+1$ copies of $V$
and call them $V_0, V_1, V_2, \dots, V_{2l}$. $V_i$ is the set of
nodes at level $i$ and for each node $v$ in $V$, we call its copy in
$V_i$ as $v_i$. For $i = 1, \ldots, l$, level $2i-1$ corresponds to
the beginning of round $i$ and level $2i$ corresponds to the end of
round $i$. Level $0$ corresponds to the network at the start. Note
that the end of a particular round and the start of the next round are
represented by different levels. There are three kinds of edges in the
graph. First, for every round $i$ and every edge $(u,v) \in G_i$, we
place two directed edges with unit capacity each, one from $u_{2i-1}$
to $v_{2i}$ and another from $v_{2i-1}$ to $u_{2i}$. We call these
edges {\em broadcast edges} as they will correspond to broadcasting of
tokens; the unit capacity on each such edge will ensure that only one
token can be sent from a node to a neighbor in one round. Second, for
every node $v$ in $V$ and every round $i$, we place an edge with
infinite capacity from $v_{2(i-1)}$ to $v_{2i}$. We call these edges
{\em buffer edges} as they ensure tokens can be stored at a node from
the end of one round to the end of the next. Finally, for every node
$v \in V$ and every round $i$, we also place an edge with unit
capacity from $v_{2(i-1)}$ to $v_{2i-1}$. We call these edges as {\em
  selection edges} as they correspond to every node selecting a token
out of those it has to broadcast in round $i$; the unit capacity
ensures that in a given round a node must send the same token to all
its neighbors. Figure \ref{fig:evolution_broadcast} illustrates our
construction, and Lemma~\ref{lem:level.steiner} explains its
usefulness.

\begin{figure}[ht]
\begin{center}
\includegraphics[width=5in]{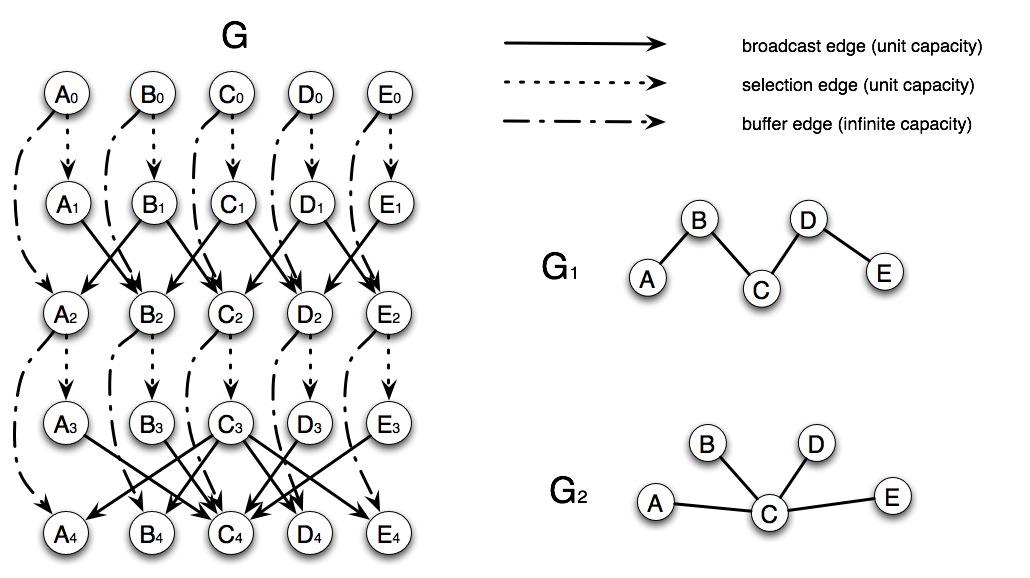}
\caption{An example of how to construct the evolution graph, for
  broadcast schedules, from a sequence of communication graphs.}
\label{fig:evolution_broadcast}
\end{center}
\end{figure}

\begin{lemma}
\label{lem:level.steiner}
Let there be $k$ tokens, each with a source and a set of
destinations. It is feasible to send all the tokens to all of their
destinations using $l$ rounds, where every node broadcasts only one
token in each round, iff $k$ directed Steiner trees can be packed in
$\bexpandedG{2l + 1}$ levels, one for each token with its root being
the copy of the source at level $0$ and its terminals being the copies
of the destinations at level $2l$.
\end{lemma}
\begin{proof}
Assume that $k$ tokens can be sent to all of their destinations in $l$
rounds and fix one broadcast schedule that achieves this. We will
construct $k$ directed Steiner trees as required by the lemma based on
how the tokens reach their destinations and then argue that they all
can be packed in $\bexpandedG{2l+1}$ respecting the edge
capacities. For a token $i$, we construct a Steiner tree $T^i$ as
follows.  For each level $j \in \{0, \ldots, 2l\}$, we define a set
$S^i_j$ of nodes at level $j$ inductively starting from level $2l$
backwards.  $S^i_{2l}$ is simply the copies of the destination nodes
for token $i$ at level $2l$. Once $S^i_{2(j+1)}$ is defined, we define
$S^i_{2j}$ (respectively $S^i_{2j+1}$) as: for each $v_{2(j+1)} \in
S^i_{2(j+1)}$, include $v_{2j}$ (respectively nothing) if token $i$
has reached node $v$ by round $j$, or include a node $u_{2j}$
(respectively $u_{2j+1}$) such that $u$ has token $i$ at the end of
round $j$ which it broadcasts in round $j+1$ and $(u,v)$ is an edge of
$G_{j+1}$. Such a node $u$ can always be found because whenever
$v_{2j}$ is included in $S^i_{2j}$, node $v$ has token $i$ by the end
of round $j$ which can be proved by backward induction staring from $j
= l$. It is easy to see that $S^i_0$ simply consists of the copy of
the source node of token $i$ at level $0$. $T^i$ is constructed on the
nodes in $\cup_{j = 0}^{j = 2l} S^i_j$. If for a vertex $v$,
$v_{2(j+1)} \in S^i_{2(j+1)}$ and $v_{2j} \in S^i_{2j}$, we add the
buffer edge $(v_{2j},v_{2(j+1)})$ in $T^i$. Otherwise, if $v_{2(j+1)}
\in S^i_{2(j+1)}$ but $v_{2j} \notin S^i_{2j}$, we add the selection
edge $(u_{2j},u_{2j+1})$ and broadcast edge $(u_{2j+1},v_{2(j+1)})$ in
$T^i$, where $u$ was the node chosen as described above. It is
straightforward to see that these edges form a directed Steiner tree
for token $i$ as required by the lemma which can be packed in
$\bexpandedG{2l+1}$. The argument is completed by noting that any unit
capacity edge cannot be included in two different Steiner trees as we
started with a broadcast schedule where each node broadcasts a single
token to all its neighbors in one round, and thus all the $k$ Steiner
trees can be simultaneously packed in $\bexpandedG{2l+1}$ respecting
the edge capacities.

Next assume that $k$ Steiner trees as in the lemma can be packed in
$\bexpandedG{2l+1}$ respecting the edge capacities. We construct a
broadcast schedule for each token from its Steiner tree in the natural
way: whenever the Steiner tree $T_i$ corresponding to token $i$ uses a
broadcast edge $(u_{2j-1},v_{2j})$ for some $j$, we let the node $u$
broadcast token $i$ in round $j$. We need to show that this is a
feasible broadcast schedule. First we observe that two different
Steiner trees cannot use two broadcast edges starting from the same
node because every selection edge has unit capacity, thus there are no
conflicts in the schedule and each node is asked to broadcast at most
one token in each round. Next we claim by induction that if node
$v_{2j}$ is in $T^i$, then node $v$ has token $i$ by the end of round
$j$. For $j = 0$, it is trivial since only the copy of the source node
for token $i$ can be included in $T^i$ from level $0$. For $j > 0$, if
$v_{2j}$ is in $T^i$, we must reach there by following the buffer edge
$(v_{2(j-1)},v_{2j})$ or a broadcast edge $(u_{2j-1},v_{2j})$. In the
former case, by induction node $v$ has token $i$ after round $j-1$
itself. In the latter case, node $u$ which had token $i$ after round
$j-1$ by induction was the neighbor of node $v$ in $G_j$ and $u$
broadcast token $i$ in round $j$, thus implying node $v$ has token $i$
after round $j$. From the above claim, we conclude that whenever a
node is asked to broadcast a token in round $j$, it has the token by
the end of round $j-1$. Thus the schedule we constructed is a feasible
broadcast schedule. Since the copies of all the destination nodes of a
token at level $2l$ are the terminals of its Steiner tree, we conclude
all the tokens reach all of their destination nodes after round $l$.
\end{proof}

\begin{figure}[ht]
\begin{center}
\includegraphics[width=5in]{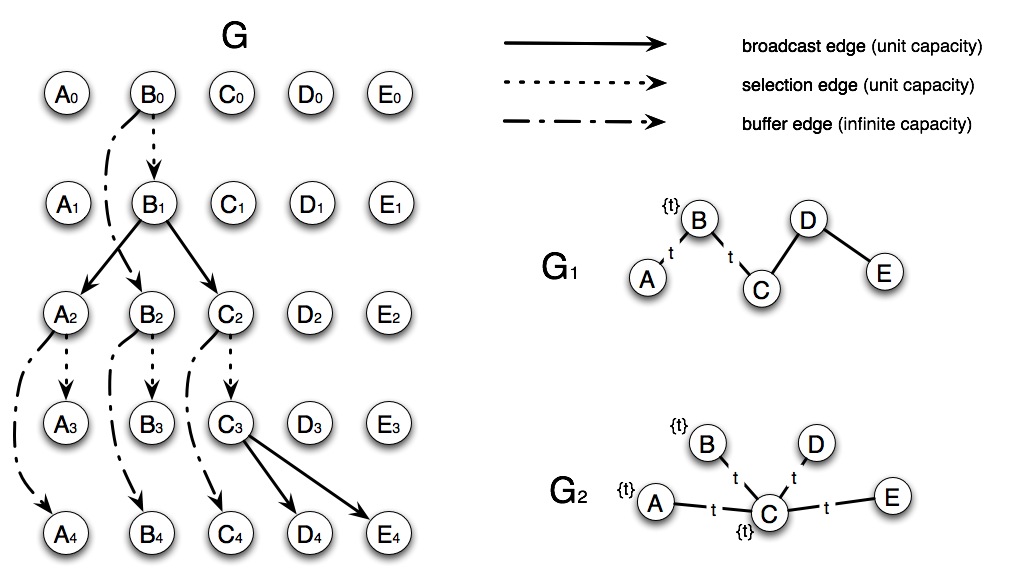}
\caption{An example of building directed Steiner tree in the evolution
  graph based on token dissemination process. Token $t$ starts
  from node $B$. Thus, the Steiner tree is rooted at $B_0$ in
  $G$. Since $B_0$ has token $t$, we include the infinite capacity
  buffer edge $(B_0,B_2)$. In the first round, node $B$ broadcasts
  token $t$, and hence we include the selection edge
  $(B_0,B_1)$. Nodes $A$ and $C$ receive token $t$ from $B$ in the
  first round, so we include edges $(B_1,A_2)$, $(B_1,C_2)$. Now
  $A_2$, $B_2$, and $C_2$ all have token $t$. Therefore we include the
  edges $(A_2,A_4)$, $(B_2,B_4)$, and $(C_2,C_4)$. In the second
  round, all of $A$, $B$, and $C$ broadcast token $t$, we include
  edges $(A_2,A_3)$, $(B_2,B_3)$, $(C_2,C_3)$. Nodes $D$ and $E$
  receive token $t$ from $C$. So we include edges $(C_3,D_4)$ and
  $(C_3,E_4)$. Notice that nodes $A$ and $B$ also receive token $t$
  from $C$, but they already have token $t$. Thus, we don't include
  edges $(C_3,B_4)$ or $(C_3,A_4)$.}
\label{fig:steiner}
\end{center}
\end{figure}

Our algorithm is given in Algorithm~\ref{alg:broadcast_flow_based} and
analyzed in Lemma~\ref{lem:level.flow} and
Theorem~\ref{thm:offline_broadcast}.

\begin{algorithm}[ht!]
\caption{$O(\min\{n \sqrt{k\log n}, nk\})$ round algorithm in the
  offline model}
\label{alg:broadcast_flow_based}
\begin{algorithmic}[1]
  \REQUIRE A sequence of communication graphs $G_i$, $i = 1, 2, \ldots$
  \ENSURE Schedule to disseminate $k$ tokens.

  \medskip

  \IF{$k \leq \sqrt{\log n}$}

  \FOR{each token $t$} \label{alg.step:flow_based.trivial}

  \STATE For the next $n$ rounds, let every node that has token
  $t$ broadcast the token.

  \ENDFOR 

  \ELSE

  \STATE Choose a set $S$ of $2\sqrt{k \log n}$ random nodes. \label{alg.step:random}
  
  \FOR{each vertex in $v \in S$} \label{alg.step:flow_based.phase_1}

  \STATE Send each of the $k$ tokens to vertex $v$ in $O(n)$ rounds. 

  \ENDFOR

  \FOR{each token $t$} \label{alg.step:flow_based.phase_2}

  \STATE For the next $2n \sqrt{(\log n)/k}$ rounds, let every node with token
  $t$ broadcast it.

  \ENDFOR

  \ENDIF

\end{algorithmic}
\end{algorithm}

\begin{lemma}
\label{lem:level.flow.broadcast}
Let $k \leq n$ tokens be at given source nodes and $v$ be an arbitrary
node. Then, all the tokens can be gathered at $v$ in the broadcast
model in at most $n+k$ rounds.
\end{lemma}

The proof is analogous to that for the multiport model and is omitted.

\junk{
\begin{theorem}
\label{thm:flow_based}
Algorithm~\ref{alg:flow_based} solves the $k$-gossip problem using
$O(\min\{n \sqrt{k \log n}, nk\})$ rounds with high probability in
the offline model.
\end{theorem}
}

\tOfflineBroadcast
\begin{proof}
It is trivial to see that if $k \leq \sqrt{\log n}$, then the
algorithm will end in $nk$ rounds and each node receives all the $k$
tokens. Assume $k > \sqrt{\log n}$. By Lemma~\ref{lem:level.flow}, all
the tokens can be sent to all the nodes in $S$ using $O(n \sqrt{k \log
  n})$ rounds. Now fix a node $v$ and a token $t$. Since token $t$ is
broadcast for $2n \sqrt{(\log n)/k}$ rounds, there is a set $S^t_v$ of
at least $2n \sqrt{(\log n)/k}$ nodes from which $v$ is reachable
within those rounds.  It is clear that if $S$ intersects $S^t_v$, $v$
will receive token $t$. Since the set $S$ was picked uniformly at
random, the probability that $S$ does not intersect $S^t_v$ is at most
\[ \frac{{n - 2n\sqrt{(\log n)/k} \choose 2\sqrt{k \log n}}}{{n \choose 2\sqrt{k \log n}}} < \left(\frac{n - 2n\sqrt{(\log n)/k}}{n}\right)^{2\sqrt{k \log n}} 
\le \frac{1}{n^4}.\]  
Thus every node receives every token with probability $1-1/n^3$. It is
also clear that the algorithm finishes in $O(n \sqrt{k \log n})$
rounds.
\end{proof}

Algorithm~\ref{alg:flow_based} can be derandomized using the technique
of conditional expectations, as shown in
Algorithm~\ref{alg:derandomize} and analyzed in
Lemma~\ref{lem:derandomize}.

Algorithm~\ref{alg:flow_based} can be derandomized using the standard
technique of conditional expectations, as shown in
Algorithm~\ref{alg:derandomize}.  Given a sequence of communication
graphs, if node $u$ broadcasts token $t$ for $\Delta$ rounds and every
node that receives token $t$ also broadcasts $t$ during that period,
then we say node $v$ is within $\Delta$ {\em broadcast distance} to
$u$ if and only if $v$ receives token $t$ by the end of round
$\Delta$. Let $S$ be a set of nodes, and $|S|\le 2 \sqrt{k \log
  n}$. We use $\dprob{u}{S}{T}$ to denote the probability that the
broadcast distance from node $u$ to set $X$ is greater than $2n
\sqrt{(\log n)/k}$, where $X$ is the union of $S$ and a set of
$2\sqrt{k\log n} - |S|$ nodes picked uniformly at random from $V
\setminus T$, and $\dsumprob{S}{T}$ denotes the sum, over all $u$ in
$V$, of $\dprob{u}{S}{T}$.

\begin{algorithm}[ht!]
\caption{Derandomized algorithm for Step~\ref{alg.step:random} in
  Algorithm~\ref{alg:flow_based}}
\label{alg:derandomize}
\begin{algorithmic}[1]
  \REQUIRE A sequence of communication graphs $G_i$, $i = 1, 2,
  \ldots$, and $k \ge \sqrt{\log n}$

  \ENSURE A set of $2\sqrt{k\log n}$ nodes $S$ such that the broadcast
  distance from every node $u$ to $S$ is within $2 n\sqrt{(\log
    n)/k}$.
  \medskip

  \STATE Set $S$ and $T$ be $\emptyset$.

  \FOR{each $v\in V$}

  \STATE $T = T \cup \{v\}$
  
  \IF{$\dsumprob{S \cup \{v\}}{T} \le \dsumprob{S}{T}$ \label{alg.step:cal}}

  \STATE $S = S\cup \{v\}$

  

  \ENDIF

  \ENDFOR

  \STATE Return $S$

\end{algorithmic}
\end{algorithm}

\begin{lemma}
\label{lem:derandomize}
The set $S$ returned by Algorithm~\ref{alg:derandomize} contains at
most $2\sqrt{k\log n}$ nodes, and the broadcast distance from every
node to $S$ is at most $2n\sqrt{(\log n)/k}$.
\end{lemma}
\begin{proof}
Let us view the process of randomly selecting $2\sqrt{k\log n}$ nodes
as a computation tree. This tree is a complete binary tree of height
$n$. There are $n+1$ nodes on any root-leaf path. The level of a node
is its distance from the root. The computation starts from the
root. Each node at the $i$th level is labeled by $b_i \in \{0,1\}$,
where 0 means not including node $i$ in the final set and 1 means
including node $i$ in the set. Thus, each root-leaf path, $b_1b_2\dots
b_n$, corresponds to a selection of nodes.  For a node $a$ in the
tree, let $S_a$ (resp., $T_a$) denote the sets of nodes that are
included (resp., lie) in the path from root to $a$.

By Theorem~\ref{thm:offline_broadcast}, we know that for the root node $r$,
we have $\dsumprob{\emptyset}{S_r} =\dsumprob{\emptyset}{\emptyset}\le
1/n^3$.  If $c$ and $d$ are the children of $a$, then $T_c$ = $T_d$,
and there exists a real $0 \le p \le 1$ such that for each $u$ in $V$,
$\dprob{u}{S_a}{T_a}$ equals $p \dprob{u}{S_c}{T_c} +
(1-p)\dprob{u}{S_d}{T_d}$.  Therefore, $\dsumprob{S_a}{T_a}$ equals $p
\dsumprob{S_c}{T_c} + (1-p) \dsumprob{S_d}{T_d}$.  We thus obtain that
$\min\{\dsumprob{S_c}{T_c},\dsumprob{S_d}{T_d}\} \le
\dsumprob{S_a}{T_a}$.  Since we set $S$ to be $X$ in $\{S_c, S_d\}$
that minimizes $\dsumprob{X}{T_c}$, we maintain the invariant that
$\dsumprob{S}{T} \le 1/n^3$.  In particular, when the algorithm
reaches a leaf $l$, we know $\dsumprob{S_l}{V}\le 1/n^3$.  But a leaf
$l$ corresponds to a complete node selection, so that
$\dprob{u}{S_l}{V}$ is 0 or 1 for all $u$, and hence
$\dsumprob{S_l}{V}$ is an integer.  We thus have $\dsumprob{S_l}{V} =
0$, implying that the broadcast distance from node $u$ to set $S_l$ is
at most $2n \sqrt{(\log n)/k}$ for every $l$.  Furthermore, $|S_l|$ is
$2 k \sqrt{\log n}$ by construction.

Finally, note that Step~\ref{alg.step:cal} of
Algorithm~\ref{alg:derandomize} can be implemented in polynomial time,
since for each $u$ in $V$, $\dprob{u}{S}{T}$ is simply the ratio of
two binomial coefficients with a polynomial number of bits.  Thus,
Algorithm~\ref{alg:derandomize} is a polynomial time algorithm with
the desired property.
\end{proof}
\fi

%% file: models.tex
\section{Models for $k$-Gossip Problem}
\label{sec:models} 

The $k$-gossip problem in dynamic networks ia a fundamental problem in
distributed computing and is rich in terms of future research
directions. It can be studied in various models, with varying
difficulty and differing along different dimensions. We present a
discussion of the most important and interesting models. We structure
the discussion based on the different dimensions along which these
models differ, which also illustrates their power and weaknesses.

One of the most important dimension for k-gossip problems is the
adversarial model used. In general, we can consider three different
types of adversaries: {\em adaptive}, {\em oblivious} and {\em
  offline}.  An adaptive adversary can adapt to the steps of the
algorithm, and in particular, base its decisions on the current state
of token distribution while laying out the network. An oblivious
adversary, on the other hand, is required to lay out the entire
network sequece before the start of the protocol, which, however, is
revealed to the algorithm one at a time in successive rounds. The
above two adversaries are meaningful in the online setting of the
problem. An offline adversary, in contrast, not only lays out the
entire network sequence in advance, but this information is also
available to the algorithm before it starts.

The adaptive adversarial model can further be subdivided as strong,
intermediate or weak based on the order of execution of the steps of
the adversary and the algorithm in each round. In the strong adaptive
addversarial model, in each round and for each node, the algorithm is
first required to decide which token to broadcst from the set of
tokens it has obtained by the end of the previous round. The adversary
then lays out the network for the current round with the complete
knowledge of the token distribution till the end of the previous round
as well as all the choices made by the algorithm for the current
round. This is the strongest type of adversary and is the first model
studied in this paper. In contrast, in the weak adaptive adversarial
model, the adversary is first required to lay down the network for the
current round with the knowledge of the token distribution till the
end of the previous round, and this network is revealed to the
algorithm while making its decisions for the current round. In the
intermediate adaptive adversarial model, the adversary and the
algorithm are required to execute their steps in parallel. That is,
the adversary is required to lay down the network with the knowledge
of the token distribution till the end of the previous round but this
network is not revealed to the algorithm while making its choices for
the current round. This kind of adversary is intermediate between
strong and weak in its power, hence the name.

The oblivious adversarial model can be further classified as strong
or weak. While for both, the adversary lays out the entire netwrok
sequence in advance of the start of the protocol, the two differs in
when the algorithm is revealed the network for the current round. In
the case of strong oblivious adversarial model, in every round, the
algorithm is first decides which token to broadcast for each node from
the set of tokens it has till the end of the previous round. The
network for the current round is then revealed to it. In the weak
oblivious adversarial model, the network for the current round is
shown to the algorithm while making its decisions for the current
round.

Another dimension in which models for the $k$-gossip problem differ is
the broadcast Vs. multi-port model. In the broadcast model, every node
broadcast at most one token in each round which is received by all of
its neighbors. In contrast, the multi-port model allows each node to
send different tokens to different neighbors. Another dimension is the
use of randomness - the models can allow randomized algorithms or
restrict to deterministic ones.

We now list which combinations of different model dimensions make
sense. In both of the strong and intermediate adaptive adversarial
models, only broadcast algorithms makes sense and we can have both
deterministic or randomied kinds. It is worth noting that the strong
adaptive and intermediate adaptive models are the same when restricted
to deterministic algorithms, as the intermediate adaptive adversary
can always compute the decisions made by the algorithm in the current
round which essentially makes it a strong adaptive adversary. In the
weak adaptive adversarial model, both broadcast and multi-port
algorithms make sense, and both of the kind can be either deterministic
or randomized.

In the strong oblivious adversarial model, only broadcast algorithms
make sense which can be either deterministic or randomized. In the
weak oblivious adversarial model and the offlie adversarial model, we
can have broadcast or multi-port algorithms and each kind can be
either deterministic or randomized. 

Our lower bound holds in the strong adaptive adversarial model against
deterministic as well as randomized broadcast algorithms. By the
equivalence between strong adaptive and intermediate adaptive models
for deterministic algorithms stated above, our lower bounds also
extend to intermediate adaptive adversarial model against
deterministic broadcast algorithms. We present a randomized multi-port
algorithm in the weak adaptive adversarial model where we start from a
well-mixed token distribution and assume the ability of $O(\log n)$
communication steps per round.

%% file: open.tex
\section{Concluding remarks and open questions}
We studied the fundamental $k$-gossip problem in dynamic networks and
showed a lower bound of $\Omega(n + nk/\log n)$ rounds for any token
forwarding algorithm against a strongly adaptive adversary,
significantly improving over the previous best bound of $\Omega(n\log
k)$~\cite{kuhn+lo:dynamic} for sufficiently large $k$.  Our lower
bound matches the known upper bound of $O(nk)$ up to a logarithmic
factor, and establishes a near-linear factor separation between
token-forwarding and network-coding based algorithms.  While our bound
rules out significantly faster algorithms in the strongly adaptive
adversary model, we complement our lower bound by presenting the
\symdiff\ protocol for a weakly adaptive adversary.  We show that
\symdiff\ is near-optimal when the starting distribution is
well-mixed.  Intuitively, a well-mixed distribution captures the
``hard" regime for information spreading in the adversarial setting,
when most nodes have most of the tokens.  Perhaps, the most
interesting problem left open by our work is the analysis of
\symdiff\ in the weakly adaptive adversary model for an {\em
  arbitrary}\/ starting distribution. \junk{ A main conjecture of our
  paper is that the \symdiff protocol gives near-linear time for {\em
    any} starting distribution.}  \junk{An interesting feature of
  \symdiff protocol is that it needs techniques from communication
  complexity for efficient implementation.}

\junk{
Our almost-tight lower bound of $\Omega(n + nk/\log n)$ rounds extends
to randomized algorithms with an adaptive adversary that makes its
decision in each round with knowledge of the randomness of the
algorithm in that round (but without knowledge of future randomness).
}

We also presented offline algorithms for $k$-gossip.  An important
intermediate model between the offline setting and the adaptive
adversary models is the oblivious adversary model in which the
adversary lays the dynamic network in advance (as in the offline
setting), but the changing topology is revealed to the algorithm one
round at a time.  Finally, this paper has focused on models in which
at most one token is sent per edge per round and the network can
change every round.  Subsequent to the announcement of our lower
bound~\cite{arxiv}, the argument has been extended to the model where
multiple tokens can be broadcast and the dynamic network is required
to contain a stable subgraph for multiple rounds~\cite{personal}.

\junk{For the important practical case of small token sizes (e.g., $O(\log
n)$ bits) even the best online gossip algorithm we know based on
network coding takes $O(n^2/\log n)$ rounds~\cite{haeupler+k:dynamic}.
In contrast, we show that in the offline setting there exist
 token-forwarding algorithms that run in
$O(n^{1.5}\sqrt{\log n})$ time.  What is the best possible time
achievable for gossip with small token sizes?}

\junk{
Finally, a major question is whether we can design fully distributed
fast (e.g., $O(\min\{n\sqrt{\log k}, nk\})$ time) token-forwarding
algorithms for general dynamic networks under an offline adversary or,
even better, an oblivious adversary.  We believe that our first
centralized algorithm can be made distributed and will be useful in
resolving this question.}

\junk{Finally, in a recent work \cite{rw-podc}, our Algorithm 1 has been
directly adapted to yield a distributed token forwarding algorithm for
information spreading, albeit in a {\em restricted} model of dynamic
networks where it is assumed that the spectral properties of the
networks don't change.  This distributed algorithm has subquadratic
running time only under certain conditions (e.g., the network should
have small dynamic diameter) and is slower in general than the
$O(\min\{n\sqrt{\log k}, nk\})$ round algorithm of this paper.  A
major question is whether we can design fully distributed fast (e.g.,
$O(\min\{n\sqrt{\log k}, nk\})$ time) token-forwarding algorithms for
general dynamic networks under an offline adversary or, even better,
an oblivious adversary.  We believe that the techniques introduced in
this paper will be useful in tackling this question.
}